\documentclass[journal,twoside,web]{ieeecolor}

\usepackage{textcomp}
\usepackage{color}
\usepackage{generic}
\usepackage{graphicx}
\usepackage{subcaption}
\usepackage{scalerel}

\usepackage{siunitx}
\usepackage{pifont}
\usepackage{graphicx}
\usepackage{amsmath,amssymb,textcomp}
\usepackage{algorithm} 
\usepackage{cite}

\usepackage{booktabs}
\usepackage{url}
\usepackage{algpseudocode}
\usepackage{paralist}
\usepackage{stmaryrd}
\usepackage{epsfig} 
\usepackage{bm} 
\usepackage{listings,lstautogobble}

\usepackage{epstopdf}
\usepackage{parcolumns}
\usepackage{mathtools}
\usepackage{mathrsfs}
\usepackage{xspace}
\usepackage[utf8]{inputenc}
\usepackage{calrsfs}
\usepackage{rotating,multirow}
\usepackage[compact]{titlesec}
\titlespacing{\section}{0pt}{1.5ex}{1.5ex}
\titlespacing{\subsection}{0pt}{0.7ex}{0.7ex}
\titlespacing{\subsubsection}{0pt}{0.4ex}{0.4ex}

\usepackage[T1]{fontenc}
\usepackage[utf8]{inputenc}

\newtheorem{definition}{Definition}
\newtheorem{proposition}{Proposition}
\newtheorem{theorem}{Theorem}
\newtheorem{remark}{Remark}
 \newtheorem{lemma}{Lemma} 
 \newtheorem{assumption}{Assumption}

  \renewcommand{\footnotesize}{\fontsize{8pt}{11pt}\selectfont}
\usepackage{authblk}

\DeclareMathAlphabet{\mathcal}{OMS}{cmsy}{m}{n}

\DeclarePairedDelimiter{\abs}{\lvert}{\rvert}
\DeclarePairedDelimiter{\norm}{\lVert}{\rVert}

\usepackage{hyperref}

\newcommand{\R}{{\mathbb{R}}}

\newcommand{\x}{\textbf{x}}

\def\t{\top}
\newcommand{\zono}[1]{\langle #1 \rangle}
\newcommand{\bigzono}[1]{\Big\langle #1 \Big\rangle}
\newcommand{\setdef}[2][]{
	\left\{
	\ifblank{#1}{}{#1 \hspace{.1cm} \middle| \hspace{.1cm}}
	#2
	\right\}
}

\allowdisplaybreaks

\begin{document}

\title{Data-Driven Reachability Analysis \\ from Noisy Data}

\author[1,2]{Amr Alanwar}
\author[3]{Anne Koch}
\author[3]{Frank Allgöwer} 
\author[1]{Karl Henrik Johansson}
\affil[1]{KTH Royal Institute of Technology \authorcr Email: {\tt \{alanwar,kallej\}@kth.se}\vspace{1ex}}
\affil[2]{Constructor University \authorcr Email: {\tt aalanwar@constructor.university}\vspace{1ex}}
\affil[3]{Institute for Systems Theory and Automatic Control, University of Stuttgart \authorcr Email: {\tt \{anne.koch,frank.allgower\}@ist.uni-stuttgart.de} \vspace{-6ex}
} 

\maketitle


\begin{abstract}

We consider the problem of computing reachable sets directly from noisy data without a given system model. Several reachability algorithms are presented for different types of systems generating the data. First, an algorithm for computing over-approximated reachable sets based on matrix zonotopes is proposed for linear systems. Constrained matrix zonotopes are introduced to provide less conservative reachable sets at the cost of increased computational expenses and utilized to incorporate prior knowledge about the unknown system model. Then we extend the approach to polynomial systems and, under the assumption of Lipschitz continuity, to nonlinear systems. Theoretical guarantees are given for these algorithms in that they give a proper over-approximate reachable set containing the true reachable set. Multiple numerical examples and real experiments show the applicability of the introduced algorithms, and comparisons are made between algorithms. 
\end{abstract}

\section{Introduction}

 Safety-critical applications require guaranteed state inclusion in a bounded set to provably avoid unsafe sets. This motivates the need for reachability analysis which computes the union of all possible trajectories that a system can reach within a finite or infinite time when starting from a bounded set of initial states \cite{conf:reach2000ellipsoidal}. The most popular approaches for computing reachable sets are set-propagation and simulation-based techniques. Set-propagation techniques propagate reachable sets for consecutive time steps. The efficiency depends on the set representation and the technique used. For instance, Taylor series methods propagate enclosures over discrete time by constructing a Taylor expansion of the states with respect to time and bounding the coefficients \cite{conf:rigorousreachtaylor}. The resulting enclosure is then inflated by a bound on the truncation error. Other set representations are polyhedra \cite{conf:reachpoly1}, zonotopes \cite{conf:reachabilityzono}, (sparse) polynomial zonotopes \cite{conf:sparsepolyzono}, ellipsoids \cite{conf:reachellipsoidal}, and support functions \cite{conf:reachsupport}. Zonotopes have favorable properties as they can be represented compactly and are closed under Minkowski sum and linear mapping. The simulation-based approach in \cite{conf:reachsim3} over-approximates the reachable set by a collection of tubes around trajectories such that the union of these tubes provides an over-approximation of the reachable set. Sampling-based approach that utilizes random set theory is presented in  \cite{conf:reachrandomsettheory}. 
\begin{figure}
    \centering
    \includegraphics[width=0.5\textwidth]{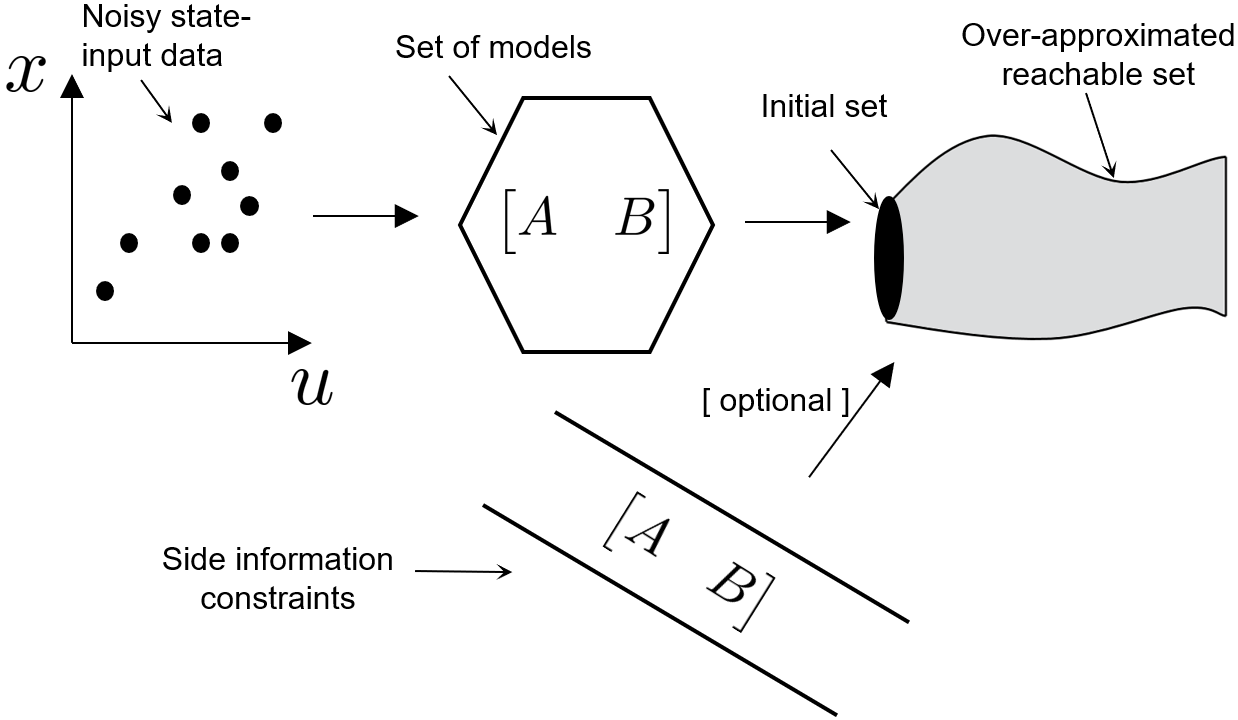}
    \caption{The reachable sets consistent with noisy input-state data are computed while making use of side information if available.}
    \label{fig:idea}
    \vspace{-4mm}
\end{figure}

 While there is a considerable amount of literature on computing reachable sets for different model classes, these methods assume an a priori given model. Obtaining a model that adequately describes the system from first principles or noisy data is usually challenging and time-consuming. Simultaneously, system data in the form of measured trajectories are often readily available in many applications. Therefore, we are interested in reachability analysis directly from noisy data of an unknown system model. One recent contribution in this direction can be found in \cite{conf:murat}, where the authors introduce two data-driven methods for computing the reachable sets with probabilistic guarantees. The first method represents the reachability problem as a binary classification problem using a Gaussian process classifier. The second method makes use of a Monte Carlo sampling approach to compute the reachable set. A probabilistic reachability analysis is proposed for general nonlinear systems using level sets of Christoffel functions in \cite{conf:murat_christoffel} where they guarantee that the algorithm's output is an accurate reachable set approximation in a probabilistic sense. 

Over-approximating reachable sets from data are considered in \cite{conf:onthefly} based on interval Taylor-based methods applied to systems with dynamics described as differential inclusions; however, the proposed approach only works under the assumption of prior nonlinear terms bound. Another interesting method is introduced in \cite{conf:tomlin_nonlin_reach}, where the model is assumed to be partially known, and data is used to learn an additional Lipschitz continuous state-dependent uncertainty, where the unknown part is assumed to be bounded by a known set. We believe that computing guaranteed reachable sets from noisy data is still an open problem. 

The main idea underlying the introduced data-driven reachability framework is visualized in Fig.~\ref{fig:idea}, where we compute data-driven reachable sets based on matrix zonotope recursion. In order to guarantee that the reachable set encloses all system trajectories from finite noisy data, we compute a matrix zonotope that encloses all models that are consistent with the noisy data instead of depending on a single model that might be incorrect. The true model is guaranteed to be within the set of models. 
We then propagate the initial set forward, utilizing this matrix zonotope to compute the reachable set. We also provide an improved reachability analysis algorithm that provides a less conservative over-approximation of the reachable set at the cost of additional computational expenses. This improved analysis scheme is enabled by introducing a new set representation, namely constrained matrix zonotopes. We utilize this novel set representation and the corresponding operations additionally to provide a systematic approach on how supplementary prior knowledge about the unknown model, like states decoupling, partial model knowledge, or given bounds on certain entries in the system matrices, can be incorporated into the reachability analysis. 
We then extend our approach to two classes of nonlinear systems: polynomial systems and Lipschitz systems. All used codes to recreate our findings are publicly available\footnotemark. 


 We will specifically build upon ideas used in \cite{conf:nonexciting,conf:dissipativity1,conf:dissipativity2} and \cite{conf:formulas,conf:annerobustcontrol,conf:berberich2020combining,conf:controlslemma} for data-driven analysis and data-driven controller design, respectively. In these works, the data is used to characterize all models that are consistent with the data. This characterization enables a computational approach for direct systems analysis and design without explicitly identifying a model. 


\footnotetext{\href{https://github.com/aalanwar/Data-Driven-Reachability-Analysis}{https://github.com/aalanwar/Data-Driven-Reachability-Analysis}}

The main contributions of this paper are as follows:
\begin{itemize}
    \item An algorithm is proposed to compute the reachable sets of an unknown control system from noise-corrupted input-state measurements using matrix zonotopes (Algorithm~\ref{alg:LTIreach}). The resultant reachable sets are shown to over-approximate the model-based reachable sets for linear time-invariant (LTI) systems (Theorem~\ref{th:reach_lin}). 
    \item A new set representation named constrained matrix zonotope (Definition~\ref{def:conmatzonotopes}) and its essential operations are proposed.
    \item The constrained matrix zonotope is exploited in Algorithm~\ref{alg:LTIConstrainedReachability} by providing less conservative reachable sets using the exact noise description. The computed reachable sets over-approximate the model-based reachable sets for LTI systems (Theorem~\ref{th:reach_lin_cmz}).
    \item A general framework is introduced for incorporating side information like states decoupling, partial model knowledge, or given bounds on certain entries in the system matrices about the unknown model (Algorithm~\ref{alg:LTISideInfoReachability}) into the reachability analysis, which further decreases the conservatism of the resulting reachable sets. The resultant reachable sets over-approximate the model-based reachable set for LTI systems (Theorem~\ref{th:sideinfo}).
    \item An algorithm is proposed for reachability analysis in the existence of process and measurement noise for LTI systems (Algorithm~\ref{alg:LTIMeasReachability}).
    \item An algorithm (Algorithm~\ref{alg:PolyReachability}) is proposed for computing the reachable sets that are guaranteed to over-approximate the exact reachable set for polynomial systems (Theorem~\ref{th:reach_lin_poly}). A variant of the LTI side information framework can be used for the polynomial systems.
    \item An algorithm is proposed for computing the reachable set  (Algorithm~\ref{alg:LipReachability}), which results in reachable sets that are guaranteed to over-approximate the exact reachable sets of nonlinear systems under a Lipschitz continuity assumption (Theorem~\ref{th:reachdisnonlin}).
    \item A real experiment along with a comparison between the proposed algorithms and the alternative direction of system identification, synthesizing a reachset
conformant model \cite{Kochdumper2020}, and model-based reachability analysis is provided.
\end{itemize} 

As discussed in more detail before, there have been a few approaches to infer the reachable sets directly from data, mostly without providing guarantees in the case of noisy data. An alternative approach is to apply well-known system identification approaches and consecutively do model-based reachability analysis. Thus, we include a comparison to a standard system identification method in Section~\ref{sec:eval} while considering $2\sigma$ uncertainty bound in the model-based reachability analysis. 
Interesting recent results on system identification with probabilistic guarantees from finite noisy data include concentration bounds \cite{conf:Matni2019}. Another very related approach is set membership estimation (see, e.g., \cite{conf:Milanese1991}), where the trade-off between conservatism and computational expenses is usually of central importance, which we will also encounter in the course of this paper.

This paper is an extension to our previous work in \cite{conf:ourL4dc}, where we introduced the basic idea of computing the reachable set by using matrix zonotopes over-approximating the set of models consistent with the data. In this work, we significantly extend and improve the ideas in \cite{conf:ourL4dc} by introducing the constrained matrix zonotope and its essential set of operations which allow the incorporation of side information about the unknown model and provide less conservative reachable sets. We also enhance the proposed nonlinear Lipschitz reachability analysis method in \cite{conf:ourL4dc}. Furthermore, we propose a new approach for computing the reachable sets of polynomial systems, and we consider measurement noise for linear systems. Our initial ideas in \cite{conf:ourL4dc} have been utilized in different applications like data-driven predictive control \cite{conf:zpc} and set-based estimation \cite{conf:setbaseddatareach}. 


The rest of the paper is organized as follows: the problem statement and preliminaries are defined in Section~\ref{sec:pb}. The new set representation (constrained matrix zonotope) is proposed in Section~\ref{sec:cmz}. Data-driven reachability analysis for LTI systems is proposed in Section~\ref{sec:reachlineardis}, including a framework to include prior knowledge into the analysis approach. Then, we extend the proposed approach to nonlinear systems in Section~\ref{sec:reachnonlinear}. The introduced approaches are applied to multiple numerical examples and experiments in Section~\ref{sec:eval}, and Section~\ref{sec:con} concludes the paper.

\section{Problem Statement and Preliminaries} \label{sec:pb}

We start by defining some set representations that are used in the reachability analysis. 

\subsection{Set Representations}
We define the following sets:

\begin{definition}(\textbf{Zonotope} \cite{conf:zono1998}) \label{def:zonotopes} 
Given a center $c_{\mathcal{Z}} \in \mathbb{R}^{n_x}$ and $\gamma_{\mathcal{Z}} \in \mathbb{N}$ generator vectors in a generator matrix $G_{\mathcal{Z}}=\begin{bmatrix} g_{\mathcal{Z}}^{(1)}& \dots &g_{\mathcal{Z}}^{(\gamma_{\mathcal{Z}})}\end{bmatrix} \in \mathbb{R}^{n_x \times \gamma_{\mathcal{Z}}}$, a zonotope is defined as
\begin{equation}
	\mathcal{Z} = \Big\{ x \in \mathbb{R}^{n_x} \; \Big| \; x = c_{\mathcal{Z}} + \sum_{i=1}^{\gamma_{\mathcal{Z}}} \beta^{(i)} \, g^{(i)}_{\mathcal{Z}} \, ,
	-1 \leq \beta^{(i)} \leq 1 \Big\} \; .
\end{equation}
We use the shorthand notation $\mathcal{Z} = \zono{c_{\mathcal{Z}},G_{\mathcal{Z}}}$ for a zonotope. 
\end{definition}
Let $L \in \mathbb{R}^{m \times n_x}$ be a linear map. Then $L\mathcal{Z}= \zono{Lc_{\mathcal{Z}},LG_{\mathcal{Z}}}$ \cite[p.18]{conf:thesisalthoff}.
Given two zonotopes $\mathcal{Z}_1=\langle c_{\mathcal{Z}_1},G_{\mathcal{Z}_1} \rangle$ and $\mathcal{Z}_2=\langle c_{\mathcal{Z}_2},G_{\mathcal{Z}_2} \rangle$, the Minkowski sum $\mathcal{Z}_1 \oplus \mathcal{Z}_2 = \{z_1 + z_2| z_1\in \mathcal{Z}_1, z_2 \in \mathcal{Z}_2 \}$ can be computed exactly as follows \cite{conf:zono1998}: 
\begin{equation}
     \mathcal{Z}_1 \oplus \mathcal{Z}_2 = \Big\langle c_{\mathcal{Z}_1} + c_{\mathcal{Z}_2}, [G_{\mathcal{Z}_1}, G_{\mathcal{Z}_2} ]\Big\rangle.
     \label{eq:minkowski}
\end{equation}
For simplicity, we use the notation $+$ instead of $\oplus$ to denote the Minkowski sum as the type can be determined from the context. Similarly, we use  $\mathcal{Z}_1 - \mathcal{Z}_2$ to denote $\mathcal{Z}_1 + -1 \mathcal{Z}_2$, not the Minkowski difference. We define and compute the Cartesian product of two zonotopes $\mathcal{Z}_1 $ and $\mathcal{Z}_2$ by 
\begin{align}\label{eq:cart}
\mathcal{Z}_1 \times \mathcal{Z}_2 &= \bigg\{ \begin{bmatrix}z_1 \\ z_2\end{bmatrix} \bigg| z_1 \in \mathcal{Z}_1, z_2 \in \mathcal{Z}_2 \bigg\} \nonumber\\
&= \Bigg\langle \begin{bmatrix} c_{\mathcal{Z}_1} \\ c_{\mathcal{Z}_2} \end{bmatrix}, \begin{bmatrix} G_{\mathcal{Z}_1} & 0 \\ 0 & G_{\mathcal{Z}_2}\end{bmatrix} \Bigg\rangle.
\end{align}

\begin{definition}\label{def:matzonotopes}(\textbf{Matrix Zonotope} \cite[p.52]{conf:thesisalthoff})  
Given a center matrix $C_{\mathcal{M}} \in \mathbb{R}^{n_x \times p}$ and $\gamma_{\mathcal{M}} \in \mathbb{N}$ generator matrices $\tilde{G}_{\mathcal{M}}=\begin{bmatrix} G_{\mathcal{M}}^{(1)}&\dots&G_{\mathcal{M}}^{(\gamma_{\mathcal{M}})} \end{bmatrix} \in \mathbb{R}^{n_x \times (p \gamma_{\mathcal{M}})}$, a matrix zonotope is defined as
\begin{equation}
	\mathcal{M} {=} \Big\{ X \in \mathbb{R}^{n_x \times p} \; \Big| \; X {=} C_{\mathcal{M}} + \sum_{i=1}^{\gamma_{\mathcal{M}}} \beta^{(i)} \, G_{\mathcal{M}}^{(i)} \, ,
	-1 \leq \beta^{(i)} \leq 1 \Big\} .
\end{equation}
We use the shorthand notation $\mathcal{M} = \zono{C_{\mathcal{M}},\tilde{G}_{\mathcal{M}}}$ for a matrix zonotope. 
\end{definition}


Zonotopes have been extended in \cite{conf:const_zono} to represent arbitrary convex polytopes by applying constraints on the $\beta$ factors. 

\begin{definition}\label{df:contzono}
(\textbf{Constrained Zonotope} \cite[Prop. 1]{conf:const_zono}) An $n_x$-dimensional constrained zonotope is defined by
\begin{equation}\label{eq:conszono}
    \mathcal{C} = \setdef[x\in\mathbb{R}^{n_x}]{x=c_{\mathcal{C}}+G_{\mathcal{C}} \beta, \ A_{\mathcal{C}} \beta=b_{\mathcal{C}}, \, \norm{\beta}_\infty\leq 1}, 
\end{equation}
where $c_{\mathcal{C}} \in \R^{n_x}$ is the center, $G_{\mathcal{C}}$ $\in$ $\R^{n_x \times n_g}$ is the generator matrix and $A_{\mathcal{C}} \in $ $\R^{n_c \times n_g}$ and $b_{\mathcal{C}} \in \R^{n_c}$ denote the constraints. In short, we use the shorthand notation $\mathcal{C}= \zono{c_{\mathcal{C}},G_{\mathcal{C}},A_{\mathcal{C}},b_{\mathcal{C}}}$ for a constrained zonotope.
\end{definition}

The main advantage of constrained zonotopes compared to polyhedral sets is that constrained zonotopes inherit the excellent scaling properties of zonotopes for increasing state-space dimensions since they are also based on a generator representation for sets \cite{conf:cora}.


\subsection{Problem Statement}
We consider a discrete-time system
\begin{align}
\begin{split}
    x(k+1) &= f(x(k),u(k))+ w(k),\\
    y(k) &= x(k) + v(k).
\end{split}
    \label{eq:sysnonlingen}
\end{align}
where $f:\mathbb{R}^{n_x}\times\mathbb{R}^{n_u} \rightarrow \mathbb{R}^{n_x}$ a twice differentiable unknown function, $w(k) \in \mathcal{Z}_w \subset \mathbb{R}^{n_x}$ denotes the noise bounded by a noise zonotope $\mathcal{Z}_w$, ${u(k) \in \mathcal{U}_k \subset \mathbb{R}^{n_u}}$ the input bounded by an input zonotope $\mathcal{U}_k$, $y(k)  \in \mathbb{R}^{n_x}$ the measured state that is additionally corrupted by measurement noise ${v(k) \in \mathcal{Z}_v \subset \mathbb{R}^{n_x}}$ bounded by the measurement noise zonotope $\mathcal{Z}_v$, and ${x(0) \in \mathcal{X}_0 \subset \mathbb{R}^{n_x}}$ the initial state of the system bounded by the initial set $\mathcal{X}_0$. 

Reachability analysis computes the set of states $x(k)$ which can be reached given a set of uncertain initial states $\mathcal{X}_0$ and a set of uncertain inputs $\mathcal{U}_k$. More formally, it can be defined as follows:
\begin{definition} (\textbf{Exact Reachable Set})
The exact reachable set $\mathcal{R}_{N}$ after $N$ time steps subject to inputs ${u(k) \in \mathcal{U}_k}$, $\forall k {=}\{ 0, \dots, N-1\}$, and noise $w( \cdot) \in \mathcal{Z}_w$, is the set of all states trajectories starting from initial set $\mathcal{X}_0$ after $N$ steps: 
\begin{align} \label{eq:R}
        \mathcal{R}_{N} = \big\{& x(N) \in \mathbb{R}^{n_x} \, \big| x(k{+}1) = f(x(k),u(k)) + w(k), \nonumber\\
       & \, x(0) \in \mathcal{X}_0,
        u(k) \in \mathcal{U}_k, w(k) \in \mathcal{Z}_w: \nonumber\\ & \forall k \in \{0,...,N{-}1\}\big\}.
\end{align}
\end{definition}
We aim to compute an over-approximation of the exact reachable sets when the model of the system in \eqref{eq:sysnonlingen} is unknown, but input and noisy state trajectories are available. More specifically, we aim to compute data-driven reachable sets in the following cases:
\begin{enumerate}
    \item LTI systems in Subsection~\ref{sec:reachlin_processnoise}:
\begin{equation*}
        x(k+1) = A_{\text{tr}} x(k) + B_{\text{tr}} u(k)+  w(k).
\end{equation*}
 where $\begin{bmatrix} A_{\text{tr}} & B_{\text{tr}} \end{bmatrix}$ denotes the true system model. 
    \item LTI systems given additional side information about the unknown model in Subsection~\ref{sec:sideinfo}. 
    \item LTI systems with measurement noise in Subsection~\ref{sec:measnoise}: 
    \begin{align*}
\begin{split}
    x(k+1) &= A_{\text{tr}} x(k) + B_{\text{tr}} u(k) + w(k),\\
     y(k) &= x(k) + v(k).
    \end{split}
\end{align*}
    \item Polynomial systems in Subsection~\ref{sec:poly}.
    \item Lipschitz nonlinear systems in Subsection~\ref{sec:libs}.
\end{enumerate}

Instead of having access to a mathematical model of the system, we consider $K$ input-state trajectories of different lengths $T_i+1$, denoted by $\{u^{(i)}(k)\}_{k=0}^{T_i - 1}$, and $\{x^{(i)}(k)\}_{k=0}^{T_i}$, $i=1, \dots, K$. We collect the set of all data sequences in the following matrices
\begin{align*}
     X &=  \begin{bmatrix} x^{(1)}(0) \dots  x^{(1)}(T_1)  \dots  x^{(K)}(0) \dots  x^{(K)}(T_K)\end{bmatrix}.
 \end{align*}
Let us further denote the shifted signals
 \begin{align*}
     X_+ &= \begin{bmatrix} x^{(1)}(1)\dots  x^{(1)}(T_1) \dots  x^{(K)}(1)  \dots  x^{(K)}(T_K) \end{bmatrix}, \nonumber\\
     X_- &= \begin{bmatrix} x^{(1)}(0) \dots  x^{(1)}(T_1\!-\!1)  \dots  x^{(K)}(0)  \dots  x^{(K)}(T_K\!-\!1) \end{bmatrix},\\
    U_- &= \begin{bmatrix} u^{(1)}(0)  \dots  u^{(1)}(T_1\!-\!1) \dots u^{(K)}(0)  \dots  u^{(K)}(T_K\!-\!1) \end{bmatrix}.
 \end{align*}
The total amount of data points from all available shifted signals is denoted by $T = \sum_{i=1}^{K} T_i$, and we denote the set of all available data by $D=(U_-,X)$. Note that when dealing with measurement noise, we will consider the trajectories $\{y^{(i)}(k)\}_{k=0}^{T_i}$ instead of $\{x^{(i)}(k)\}_{k=0}^{T_i}$. 

\subsection{Noise Zonotope and Notations}
We denote the unknown process noise in state trajectory $i$ by $\hat{w}^{(i)}(\cdot)$. It follows directly that the stacked matrix of the noise $\hat{w}^{(i)}(k)$ in the collected data:
\begin{align*}
    \hat{W}_- = \begin{bmatrix} 
    \hat{w}^{(1)}(0) \dots  \hat{w}^{(1)}(T_1\!-\!1) \dots  \hat{w}^{(K)}(0)  \dots \hat{w}^{(K)}(T_K\!-\!1)\end{bmatrix}
\end{align*}
is an element of the set $\mathcal{M}_w$ where $\mathcal{M}_w =\zono{ C_{\mathcal{M}_w},\tilde{G}_{\mathcal{M}_w} }$, with 
\begin{align}
\tilde{G}_{\mathcal{M}_w}=\begin{bmatrix} G_{\mathcal{M}_w}^{(1)}&\dots&G_{\mathcal{M}_w}^{(\gamma_{\mathcal{M}_w})}\end{bmatrix}. \label{eq:Gtildemw}
\end{align}
Note that $\mathcal{M}_w$ is the matrix zonotope resulting from the concatenation of multiple noise zonotopes $\mathcal{Z}_w=\zono{c_{\mathcal{Z}_w},\begin{bmatrix} g_{\mathcal{Z}_w}^{(1)}& \dots & g_{\mathcal{Z}_w}^{(\gamma_{\mathcal{Z}_w})}\end{bmatrix}}$ as follows:
\begin{align}
    C_{\mathcal{M}_w} &= \begin{bmatrix}c_{\mathcal{Z}_w} & \dots & c_{\mathcal{Z}_w}\end{bmatrix}, \label{eq:Cmw} \\
    \begin{split}
    \label{eq:Gmw}
    G^{(1+(i-1)T)}_{\mathcal{M}_w} &= \begin{bmatrix} g_{\mathcal{Z}_w}^{(i)} & 0_{n \times  (T-1)}\end{bmatrix}, \\
    G^{(j+(i-1)T)}_{\mathcal{M}_w} &= \begin{bmatrix} 0_{n \times  (j-1)} &g_{\mathcal{Z}_w}^{(i)}  & 0_{n \times  (T-j)}\end{bmatrix}, \\
    G^{(T+(i-1)T)}_{\mathcal{M}_w} &= \begin{bmatrix} 0_{n \times (T-1)}& g_{\mathcal{Z}_w}^{(i)}\end{bmatrix}.
    \end{split}
\end{align}
$\forall i =\{1, \dots, \gamma_{\mathcal{Z}_w}\}$, $j=\{2,\dots,T-1\}$. The set of real and natural numbers are denoted as $\mathbb{R}$ and $\mathbb{N}$, respectively, and $\mathbb{N}_0 = \mathbb{N}\cup \{0\}$. The transpose and Moore-Penrose pseudoinverse of a matrix $X$ are denoted as $X^\t$ and $X^\dagger$, respectively. We denote the Kronecker product by $\otimes$. We denote the element at row $i$ and column $j$ of matrix $A$ by $(A)_{i,j}$ and column $j$ of $A$ by $(A)_{.,j}$. 
For a list or vector of elements, we denote the element $i$ of vector or list $a$ by $a^{(i)}$. 
For a given matrix $A$, $A^0_{(i,j)}$ denotes a matrix of same size as A with zero entries everywhere except for the value $(A)_{i,j}$ at row i and column j. The vectorization of a matrix $A$ is defined by $\text{vec}(A)$. The element-wise multiplication of two matrices is denoted by $\odot$. 
We denote the over-approximation of a reachable set $\mathcal{R}_k$ by an interval by $\text{int}(\mathcal{R}_k)$. We define also for $N$ time steps
\begin{align}
    \mathcal{F} = \cup_{k=0}^{N} (\mathcal{R}_k \times \mathcal{U}_k).
    \label{eq:F}
\end{align}
Finally, we denote all system matrices $\begin{bmatrix} A & B \end{bmatrix}$ that are consistent with the data $D=(U_-,X)$ by $\mathcal{N}_{\Sigma}$:
\begin{align}
    \mathcal{N}_{\Sigma} = \{ \begin{bmatrix} A & B \end{bmatrix} | \; X_+ = A X_- + B U_- + W_-, W_- \in \mathcal{M}_w \}. \label{eq:Nsig}
\end{align}
By assumption, $\begin{bmatrix} A_{\text{tr}} & B_{\text{tr}} \end{bmatrix} \in \mathcal{N}_{\Sigma}$. 




\section{Constrained Matrix Zonotope}
\label{sec:cmz}

\begin{table}[tbp]
\caption{Complexity comparison between sets with $\gamma_{\mathcal{Z}}=n_x$ and $L \in \mathbb{R}^{m \times n_x}$ \cite{conf:annualreviewAlthoff}.}
\label{tab:complexity}
  \resizebox{0.49\textwidth}{!}{%
\centering
\normalsize
\begin{tabular}{l  c c c}
\toprule
 & Polytope (V-rep.) & Polytope (H-rep.) & Zonotope\\
\midrule
Linear Map $L$ & $\mathcal{O}(mn_x2^{n_x})$  &  $\mathcal{O}(n_x^3)$ &   $\mathcal{O}(mn_x^2)$\\
Minkowski Sum &  $\mathcal{O}(n_x2^{n_x})$  & $\mathcal{O}(2^{n_x})$  & $\mathcal{O}(n_x)$ \\
Cartesian Product & $\mathcal{O}(1)$  & $\mathcal{O}(1)$  & $\mathcal{O}(1)$ \\
\bottomrule
\end{tabular}
}
\vspace{-4mm}
\end{table}

Zonotopes have advantages over polytopes in both V-representation and H-representations as summarized in Table~\ref{tab:complexity}. It is beneficial to extend the notion of matrix zonotopes to constrained matrix zonotopes as follows.

\begin{definition}\textbf{(Constrained Matrix Zonotope)} \label{def:conmatzonotopes}  
Given a center matrix $C_{\mathcal{N}} \in \mathbb{R}^{n_x \times p}$ and a number $\gamma_{\mathcal{N}} \in \mathbb{N}$ of generator matrices $\tilde{G}_{\mathcal{N}}=[G_{\mathcal{N}}^{(1)}\, \dots \,G_{\mathcal{N}}^{(\gamma_{\mathcal{N}})}] \in \mathbb{R}^{n_x \times (p  \gamma_{\mathcal{N}}) }$, as well as matrices $\tilde{A}_{\mathcal{N}}=[A_{\mathcal{N}}^{(1)}\, \dots \,A_{\mathcal{N}}^{(\gamma_{\mathcal{N}})}] \in \mathbb{R}^{n_c \times (n_a \gamma_{\mathcal{N}}) }$ and $B_{\mathcal{N}} \in \mathbb{R}^{n_c \times n_a}$ constraining the factors $\beta^{(1:\gamma_{\mathcal{N}})}$, a constrained matrix zonotope is defined by 
\begin{align*}
	\mathcal{N} = \Big\{ X \in \mathbb{R}^{n_x \times p} \; \Big|& \; X = C_{\mathcal{N}} + \sum_{i=1}^{\gamma_{\mathcal{N}}} \beta^{(i)} \, G_{\mathcal{N}}^{(i)}, \, \nonumber\\
	& \sum_{i=1}^{\gamma_{\mathcal{N}}} \beta^{(i)} A_{\mathcal{N}}^{(i)} = B_{\mathcal{N}} \,,  -1 \leq \beta^{(i)} \leq 1 \Big\} \; .
\end{align*}
Furthermore, we define the shorthand notation ${\mathcal{N} {=} \zono{C_{\mathcal{N}},\tilde{G}_{\mathcal{N}},\tilde{A}_{\mathcal{N}},B_{\mathcal{N}}}}$ for a constrained matrix zonotope. 
\end{definition}

The constrained matrix zonotopes are closed under Minkowski sum and multiplication by a scalar which can be computed as follows:

\begin{proposition}
\label{prob:minsum}
For every $\mathcal{N}_1 {=} \zono{C_{\mathcal{N}_1},\tilde{G}_{\mathcal{N}_1},\tilde{A}_{\mathcal{N}_1},B_{\mathcal{N}_1}} \subset \mathbb{R}^{n_x \times p}$, ${\mathcal{N}_2 {=} \zono{C_{\mathcal{N}_2},\tilde{G}_{\mathcal{N}_2},\tilde{A}_{\mathcal{N}_2},B_{\mathcal{N}_2}} \, {\subset} \, \mathbb{R}^{n_x \times p}}$, and $R \,{\in} \mathbb{R}^{k \times n_x}$ the following identities hold
\begin{align}
    R\mathcal{N}_1 &= \zono{RC_{\mathcal{N}_1},R\tilde{G}_{\mathcal{N}_1},\tilde{A}_{\mathcal{N}_1},B_{\mathcal{N}_1}}, \label{eq:RN} \\
    \mathcal{N}_1 {+} \mathcal{N}_2 &{=}\Bigg\langle C_{\mathcal{N}_1} + C_{\mathcal{N}_2},[\tilde{G}_{\mathcal{N}_1}, \tilde{G}_{\mathcal{N}_2} ], \tilde{A}_{\mathcal{N}_{12}}, \begin{bmatrix} B_{\mathcal{N}_1} & 0 \\ 
    0 & B_{\mathcal{N}_2}\end{bmatrix}  \Bigg\rangle,\label{eq:N1pN2} 
\end{align}
where
\begin{align*}
\tilde{A}_{\mathcal{N}_{12}}\!{=}\!\begin{bmatrix}  \begin{bmatrix} A^{(1)}_{\mathcal{N}_1} &\!\! 0 \\  0 &\!\! 0\end{bmatrix}\!\dots\! \begin{bmatrix} A^{(\gamma_{\mathcal{N}_1})}_{\mathcal{N}_1}\! &\! 0 \\  0 \!&\! 0\end{bmatrix}\!\!\begin{bmatrix} 0 & 0 \\  0 & A^{(1)}_{\mathcal{N}_2}\end{bmatrix}\!\dots\!\begin{bmatrix} 0 \!&\! 0 \\  0 \!&\! A^{(\gamma_{\mathcal{N}_2})}_{\mathcal{N}_2}\end{bmatrix} \end{bmatrix}.
\end{align*}
\end{proposition}

 A proof of Proposition~\ref{prob:minsum} is provided in the Appendix for completeness. During the propagation of the reachable sets, we additionally need to multiply the constrained matrix zonotope by a zonotope or constrained zonotope. The result of both operations can be over-approximated by a constrained zonotope. We provide these operations in the following proposition in which we multiply centers and generators of the constrained zonotope by the center and generators of the constrained matrix zonotope. Instead of finding the factor range $\beta$ that satisfies the constraints of the constrained matrix zonotope and constrained zonotope, we scale down the generators that result from the multiplications by a scalar $\bar{d}$.  
 \begin{proposition}
  \label{prop:cmzconstzonotope}
For every $\mathcal{N} = \zono{C_{\mathcal{N}},\tilde{G}_{\mathcal{N}},\tilde{A}_{\mathcal{N}},B_{\mathcal{N}}} \subset \mathbb{R}^{p \times n_x}$, and $\mathcal{C}= \zono{c_{\mathcal{C}},G_{\mathcal{C}},A_{\mathcal{C}},b_{\mathcal{C}}} \subset \mathbb{R}^{n_x}$ the following identity holds
\begin{align}
    \mathcal{N} \mathcal{C} \subseteq & \bigzono{C_{\mathcal{N}} c_{\mathcal{C}},\begin{bmatrix} \tilde{G}_{\mathcal{N}}c_{\mathcal{C}} & C_{\mathcal{N}} G_{\mathcal{C}}& G_{f} \end{bmatrix} , \nonumber\\
    &\begin{bmatrix} A_{\mathcal{N}\mathcal{C}}& 0& 0 \\
     0 & A_{\mathcal{C}}&0
   \end{bmatrix},\begin{bmatrix} \text{vec}(B_{\mathcal{N}}) \\  b_{\mathcal{C}}\end{bmatrix}  },  \label{eq:matczono}
\end{align}
where
\begin{align}
A_{\mathcal{N}\mathcal{C}}&=\begin{bmatrix} \text{vec}(A^{(1)}_{\mathcal{N}}) & \dots & \text{vec}(A_{\mathcal{N}}^{(\gamma_{\mathcal{N}})}) \end{bmatrix},\nonumber\\
G_{f} &= \begin{bmatrix} g_{f}^{(1)} & \dots & g_{f}^{(\gamma_{\mathcal{C}}\gamma_{\mathcal{N}})} \end{bmatrix}, \nonumber\\
g_{f}^{(k)} &= \bar{d}^{(k)} G^{(i)}_{\mathcal{N}} g^{(j)}_{\mathcal{C}}, \exists\ k \  \forall i =\{1, \dots, \gamma_{\mathcal{N}}\}, \; \text{and} \nonumber\\
&j = \{1,\dots,\gamma_{\mathcal{C}} \} \ \text{such that} \  k = \{1,\dots, \gamma_{\mathcal{C}}\gamma_{\mathcal{N}}\}, \label{eq:fGg2}\\
\bar{d}^{(k)} &{=} \max (\abs{\beta_{L,\mathcal{N}}^{(i)}\beta_{L,\mathcal{C}}^{(j)}},\abs{\beta_{L,\mathcal{N}}^{(i)}\beta_{U,\mathcal{C}}^{(j)}},\abs{\beta_{U,\mathcal{N}}^{(i)}\beta_{L,\mathcal{C}}^{(j)}},\abs{\beta_{U,\mathcal{N}}^{(i)}\beta_{U,\mathcal{C}}^{(j)}}), \label{eq:fmax2} \\
\beta_{L,\mathcal{N}}^{(i)} &= \min_{\beta^{(i)}} \beta^{(i)}:\  \sum_{j=1}^{\gamma_{\mathcal{N}}} \beta^{(j)} A_{\mathcal{N}}^{(j)} = B_{\mathcal{N}},\  \norm{\beta}_\infty \leq 1, \label{eq:betaLZono2}\\
\beta_{U,\mathcal{N}}^{(i)} &= \max_{\beta^{(i)}} \beta^{(i)}:\  \sum_{j=1}^{\gamma_{\mathcal{N}}} \beta^{(j)} A_{\mathcal{N}}^{(j)} = B_{\mathcal{N}},\  \norm{\beta}_\infty \leq 1, \label{eq:betaUZono2} \\
\beta_{L,\mathcal{C}}^{(j)} &= \min_{\beta^{(j)}} \beta^{(j)}:\  A_{\mathcal{C}} \beta=b_{\mathcal{C}},\  \norm{\beta}_\infty \leq 1, \label{eq:betaLCZono2}\\
\beta_{U,\mathcal{C}}^{(j)} &= \max_{\beta^{(j)}} \beta^{(j)}:\  A_{\mathcal{C}} \beta=b_{\mathcal{C}},\  \norm{\beta}_\infty \leq 1. \label{eq:betaUCZono2}
\end{align}
\end{proposition}

A proof of Proposition~\ref{prop:cmzconstzonotope} is provided in the Appendix.

\section{Data-Driven Reachability for Linear Systems}\label{sec:reachlineardis}

We consider in this section LTI systems given (i) data corrupted by process noise, (ii) data corrupted by process noise while having additional prior information on the system matrices, and (iii) data corrupted by process noise and measurement noise. 
\subsection{Linear Systems with Process Noise}\label{sec:reachlin_processnoise}
Consider a discrete-time linear system
\begin{equation}
    \begin{split}
        x(k+1) &= A_{\text{tr}} x(k) + B_{\text{tr}} u(k)+  w(k),
    \end{split}
    \label{eq:sys}
\end{equation}
where $A_{\text{tr}} \in \mathbb{R}^{n_x \times n_x}$, and $B_{\text{tr}} \in \mathbb{R}^{n_x \times n_u}$. Due to the presence of noise, there generally exist multiple matrices $\begin{bmatrix}A & B \end{bmatrix}$ that are consistent with the data. 
To provide reachability analysis guarantees, we need to consider all models that are consistent with the data. Therefore, we are interested in computing a set $\mathcal{M}_\Sigma$ that contains all possible $\begin{bmatrix}A & B \end{bmatrix}$ that are consistent with the input-state measurements and the given noise bound and we guarantee that the true model $\begin{bmatrix} A_{\text{tr}} & B_{\text{tr}} \end{bmatrix}$ is inside the set of models $\mathcal{M}_\Sigma$. We build upon ideas from \cite{conf:dissipativity1} to our zonotopic noise descriptions, which yields a \textit{matrix zonotope} $\mathcal{M}_\Sigma \supseteq \mathcal{N}_{\Sigma}$ 
paving the way to a computationally simple reachability analysis.

%
%
\begin{lemma}
\label{lm:sigmaM}
Given input-state trajectories $D = (U_-,X)$ of the system in \eqref{eq:sys} such that $\begin{bmatrix} 
    X_-^\t & U_-^\t 
    \end{bmatrix}^\t$ has a full row rank, then the matrix zonotope 
\begin{align}
    \mathcal{M}_\Sigma = (X_{+} - \mathcal{M}_w) \begin{bmatrix} 
    X_- \\ U_- 
    \end{bmatrix}^{\dagger}
    \label{eq:zonoAB}
\end{align} 
 contains all matrices $\begin{bmatrix}A & B \end{bmatrix}$ that are consistent with the data $D = (U_-,X)$ and the noise bound, i.e., $\mathcal{M}_\Sigma \supseteq \mathcal{N}_{\Sigma}$. 
\end{lemma}
\begin{proof}
For any $\begin{bmatrix}A & B \end{bmatrix} \in \mathcal{N}_\Sigma$, 
we know that there exists a $W_- \in \mathcal{M}_w$ such that 
\begin{align}
    A X_- + B U_- = X_+ - W_-.
    \label{eq:pf1_1}
\end{align}
Every $W_- \in \mathcal{M}_w$ can be represented by a specific choice $\hat{\beta}^{(i)}_{\mathcal{M}_w}$, $-1 \leq \hat{\beta}^{(i)}_{\mathcal{M}_w} \leq 1$, $i=1,\dots,\gamma_{\mathcal{M}_w}$, that results in a matrix inside the matrix zonotope $\mathcal{M}_w$:
\begin{align*}
    W_- &= C_{\mathcal{M}_w} + \sum_{i=1}^{\gamma_{\mathcal{M}_w}} \hat{\beta}^{(i)}_{\mathcal{M}_w} G_{\mathcal{M}_w}^{(i)}.
\end{align*}
Multiplying by $\begin{bmatrix} X_- \\ U_- \end{bmatrix}^{\dagger}$ from the right to both sides in \eqref{eq:pf1_1} yields
\begin{align}
   \begin{bmatrix} A & B \end{bmatrix} = \left( X_{+} -  C_{\mathcal{M}_w} -\sum_{i=1}^{\gamma_{\mathcal{M}_w}} \hat{\beta}^{(i)}_{\mathcal{M}_w} G_{\mathcal{M}_w}^{(i)} \right) \begin{bmatrix} 
    X_- \\ U_- 
    \end{bmatrix}^{\dagger}. 
   \label{eq:pf2}
\end{align}
Hence, for all $\begin{bmatrix}A & B \end{bmatrix} \in \mathcal{N}_{\Sigma}$, 
there exists $\hat{\beta}^{(i)}_{\mathcal{M}_w}$, $-1~\leq~\hat{\beta}^{(i)}_{\mathcal{M}_w}~\leq~1$, $i=1,\dots,\gamma_{\mathcal{M}_w}$, such that \eqref{eq:pf2} holds. Therefore, for all $\begin{bmatrix}A & B \end{bmatrix} \in \mathcal{N}_{\Sigma}$, it also holds that $\begin{bmatrix} A & B \end{bmatrix} \in \mathcal{M}_\Sigma$ as defined in \eqref{eq:zonoAB}, which concludes the proof.
\end{proof}
%
%
%

\begin{remark}
The condition of having the full matrix row rank, i.e., $\mathrm{rank} \begin{bmatrix}     X_-^\t & U_-^\t 
\end{bmatrix}^\t = n_x + n_u$ in Lemma~\ref{lm:sigmaM} implies that there exists a right-inverse of the matrix $\begin{bmatrix} X_-^\t & U_-^\t \end{bmatrix}^\t$. This condition can be easily checked given the data. Note that for noise-free measurements, this rank condition can also be enforced by choosing the input persistently exciting of order $n_x+1$ if the system is controllable (compare to  \cite[Cor. 2]{conf:willems}).
\end{remark}


To guarantee an over-approximation of the reachable sets for the unknown system, we need to consider the union of reachable sets of all $\begin{bmatrix}A & B \end{bmatrix}$ that are consistent with the data. We apply the results of Lemma~\ref{lm:sigmaM} and do reachability analysis to all systems in the set $\mathcal{M}_\Sigma$. Let $\hat{\mathcal{R}}_{k}$ denotes the reachable set computed based on the noisy data using matrix zonotopes. We propose Algorithm~\ref{alg:LTIreach} to compute $\hat{\mathcal{R}}_{k}$ as an over-approximation of the exact reachable set $\mathcal{R}_{k}$. The set of models that is consistent with data is computed in line~\ref{ln:algLTIMsigma} which is then utilized in the recursion of computing the reachable set $\hat{\mathcal{R}}_{k+1}$ in line~\ref{ln:algLTIRhat}. The following theorem proves that $\hat{\mathcal{R}}_{k} \supseteq \mathcal{R}_{k}$.

\begin{algorithm}[t]
  \caption{LTI-Reachability}
  \label{alg:LTIreach}
  \textbf{Input}: input-state trajectories $D = (U_-,X)$, initial set $\mathcal{X}_{0}$, process noise zonotope $\mathcal{Z}_w$ and matrix zonotope $\mathcal{M}_w$, and input zonotope $\mathcal{U}_k$, $\forall k = 0, \dots,N-1$\\
  \textbf{Output}: reachable sets $\hat{\mathcal{R}}_{k}, \forall k = 1, \dots,N$
  \begin{algorithmic}[1]
  \State $\hat{\mathcal{R}}_{0} =\mathcal{X}_{0}$
  \State $\mathcal{M}_\Sigma = (X_{+} - \mathcal{M}_w) \begin{bmatrix} X_- \\ U_- \end{bmatrix}^\dagger$ \label{ln:algLTIMsigma}
  \For{$k = 0:N-1$}
  \State $\hat{\mathcal{R}}_{k+1} = \mathcal{M}_{\Sigma} (\hat{\mathcal{R}}_{k} \times \mathcal{U}_{k}  ) +  \mathcal{Z}_w$ \label{ln:algLTIRhat}
  \EndFor
  \end{algorithmic}
\end{algorithm}


\begin{theorem}
\label{th:reach_lin}
Given input-state trajectories $D = (U_-,X)$ of the system in \eqref{eq:sys} such that $\begin{bmatrix} 
    X_-^\t & U_-^\t 
    \end{bmatrix}^\t$ has a full row rank, then the reachable set computed in Algorithm~\ref{alg:LTIreach} over-approximates the exact reachable set, i.e., $\hat{\mathcal{R}}_{k} \supseteq \mathcal{R}_{k}$. 
\end{theorem}

\begin{proof} The reachable set computed based on the model can be found using
\begin{align}
\mathcal{R}_{k+1} &=\begin{bmatrix} A_{\text{tr}} & B_{\text{tr}} \end{bmatrix} (\mathcal{R}_{k} \times \mathcal{U}_{k} )+  \mathcal{Z}_w.
\end{align}
Since $\begin{bmatrix} A_{\text{tr}} & B_{\text{tr}} \end{bmatrix} {\in} \mathcal{M}_{\Sigma}$ according to Lemma~\ref{lm:sigmaM} and both $\mathcal{R}_k$ and $\hat{\mathcal{R}}_k$ start from the same initial set $\mathcal{X}_0$, it holds that ${\mathcal{R}_{k+1} \subseteq \hat{\mathcal{R}}_{k+1}}$.
\end{proof}

Lemma~\ref{lm:sigmaM} provides a matrix zonotope $\mathcal{M}_\Sigma$ which comprises all $\begin{bmatrix}A & B \end{bmatrix}$ that are consistent with the data and the noise bound. However, not all elements of the matrix zonotope $\mathcal{M}_\Sigma$ correspond to a system in \eqref{eq:sys} that can explain the data given the noise bound, i.e., $\mathcal{M}_\Sigma$ is in fact a superset of all $\begin{bmatrix}A & B \end{bmatrix}$ that are consistent with the data ($\mathcal{N}_{\Sigma} \subseteq \mathcal{M}_\Sigma$). As discussed in~\cite{conf:annerobustcontrol,conf:dissipativity1}, $X_+ - W_-$ might not be explainable by $AX_-~+~BU_-$ for all possible $W_- \in \mathcal{M}_w$. More precisely, there might not exists a solution $\begin{bmatrix}A & B \end{bmatrix}$ to the system of linear equations
\begin{align*}
    \begin{bmatrix}A & B \end{bmatrix} \begin{bmatrix} X_- \\ U_- \end{bmatrix} = X_+ - W_-
\end{align*}
for all $W_- \in \mathcal{M}_w$. An exact description for all systems consistent with the data and the noise bound would therefore be the set 
\begin{align}
    \mathcal{N}_{\Sigma} = (X_{+} - \mathcal{N}_{w}) \begin{bmatrix} 
    X_- \\ U_- 
    \end{bmatrix}^{\dagger}
    \label{eq:N_Sigma}
\end{align} 
with
\begin{align}
    \mathcal{N}_{w} = \Bigg\{ W_- \in \mathcal{M}_w \; \Bigg| \; (X_+ - W_-) \begin{bmatrix} X_- \\ U_- \end{bmatrix}^\perp = 0\Bigg\},
    \label{eq:N_w}
\end{align}
where $\begin{bmatrix} X_- \\ U_- \end{bmatrix}^\perp$ denotes a matrix containing a basis of the kernel of $\begin{bmatrix} 
    X_-^\t & U_-^\t 
    \end{bmatrix}^\t$. Representing $\mathcal{N}_{w}$ and $\mathcal{N}_{\Sigma}$ is not possible using state of the art zonotopes representations. Therefore, we propose the constrained matrix zonotope introduced in Section~\ref{sec:cmz} as a new set representation that can represent the sets $\mathcal{N}_{w}$ and thereby $\mathcal{N}_{\Sigma}$ 
to compute a less conservative reachable set $\bar{\mathcal{R}}_{k}$ at the cost of increasing the computational complexity in Algorithm~\ref{alg:LTIConstrainedReachability}. Due to adding constraints, $\bar{\mathcal{R}}_{k}$ is a constrained zonotope for $k>0$ different from $\hat{\mathcal{R}}_{k}$ which is a zonotope. We first compute the exact noise description $\mathcal{N}_w$ in line~\ref{ln:algconstANw} to line~\ref{ln:algconstNw}. Then, we compute the set of models $\mathcal{N}_{\Sigma}$ that is consistent with the exact noise description in line~\ref{ln:algconstNsigma} which is further utilized in the recursion of computing the reachable set $\bar{\mathcal{R}}_{k+1}$ in line~\ref{ln:algconstRbar}. The following theorem proves that $\mathcal{R}_{k} \subseteq  \bar{\mathcal{R}}_k$.

\begin{algorithm}[t]
  \caption{LTI-Constrained-Reachability}
  \label{alg:LTIConstrainedReachability}
  \textbf{Input}: input-state trajectories $D = (U_-,X)$, initial set $\mathcal{X}_{0}$, process noise zonotope $\mathcal{Z}_w$ and matrix zonotope $\mathcal{M}_w$, and input zonotope $\mathcal{U}_k$, $\forall k = 0, \dots,N-1$\\
  \textbf{Output}: reachable sets $\bar{\mathcal{R}}_{k}, \forall k = 1, \dots,N$  
  \begin{algorithmic}[1]
  \State $\bar{\mathcal{R}}_{0} =\mathcal{X}_{0}$
  \State $A^{(i)}_{\mathcal{N}_w} = G^{(i)}_{\mathcal{M}_w} \begin{bmatrix} X_- \\ U_- \end{bmatrix}^\perp,\quad  \forall i=\{1, \dots,\gamma_{\mathcal{Z}_w} T\}$ \label{ln:algconstANw}
  \State $\tilde{A}_{\mathcal{N}_w} =\begin{bmatrix} A_{\mathcal{N}_w}^{(1)}& \dots & A_{\mathcal{N}_w}^{(\gamma_{\mathcal{Z}_w}T)} \end{bmatrix}$
  \State $B_{\mathcal{N}_w} = (X_{+} - C_{\mathcal{M}_w})\begin{bmatrix} X_- \\ U_- \end{bmatrix}^\perp$
  \State $\mathcal{N}_w = \zono{C_{\mathcal{M}_w},\tilde{G}_{\mathcal{M}_w},\tilde{A}_{\mathcal{N}_w},B_{\mathcal{N}_w}}$ \label{ln:algconstNw}
  \State $\mathcal{N}_{\Sigma} = (X_{+} - \mathcal{N}_w) \begin{bmatrix} X_- \\ U_- \end{bmatrix}^\dagger$ \label{ln:algconstNsigma}
  \For{$k = 0:N-1$} 
  \State $\bar{\mathcal{R}}_{k+1} =\mathcal{N}_{\Sigma} (\bar{\mathcal{R}}_{k} \times \mathcal{U}_{k}  ) +  \mathcal{Z}_w$ \label{ln:algconstRbar}
  \EndFor
  \end{algorithmic}
\end{algorithm}

\begin{theorem}
\label{th:reach_lin_cmz}
Given input-state trajectories $D = (U_-,X)$ of the system in \eqref{eq:sys} such that $\begin{bmatrix} 
    X_-^\t & U_-^\t 
    \end{bmatrix}^\t$ has a full row rank, then the reachable set computed in Algorithm~\ref{alg:LTIConstrainedReachability} over-approximates the exact reachable set, i.e., $\bar{\mathcal{R}}_k \supseteq \mathcal{R}_{k}$.
\end{theorem}

\begin{proof}
As pointed out in \cite{conf:annerobustcontrol}, the condition for the existence of a solution $F_{[A\ B]}$ to the system of linear equations
\begin{align*}
    F_{[A\ B]} \begin{bmatrix}
    X_- \\ U_-
    \end{bmatrix} = X_+ - W_-, 
    \end{align*}
    or equivalently
    \begin{align}
    \begin{bmatrix}
    X_- \\ U_-
    \end{bmatrix}^\t F_{[A\ B]}^\t = (X_+ - W_-)^\t
    \label{eq:fred}
\end{align}
can be reformulated via the Fredholm alternative as
\begin{align*}
    \begin{bmatrix}
    X_- \\ U_-
    \end{bmatrix} \tilde{z} = 0 \quad \Rightarrow 
    \quad (X_+ - W_-)\tilde{z} = 0,
\end{align*}
which means that any vector $\tilde{z} \in \mathbb{R}^{T}$ in the kernel of $\begin{bmatrix} 
    X_-^\t & U_-^\t 
    \end{bmatrix}^\t$ must also lie in the kernel of $X_+ - W_-$.
Since 
$\begin{bmatrix}
    X_- \\ U_-
    \end{bmatrix}^\perp$ contains a basis of the kernel of $\begin{bmatrix} 
    X_-^\t & U_-^\t 
    \end{bmatrix}^\t$, another equivalent condition for the existence of a solution $F_{[A\ B]}$ in \eqref{eq:fred} is hence
\begin{align}
    (X_+ - {W}_-) \begin{bmatrix} X_- \\ U_- \end{bmatrix}^\perp = 0.
    \label{eq:constraint}
\end{align}
Considering the constraint~\eqref{eq:constraint} together with the bounding matrix zonotope $\mathcal{M}_w = \zono{C_{\mathcal{M}_w},\tilde{G}_{\mathcal{M}_w}}$, we find:
\begin{align}
    \Bigg(X_+ - C_{\mathcal{M}_w} - \sum_{i=1}^{\gamma_{\mathcal{Z}_w} T} \beta^{(i)} G_{\mathcal{M}_w}^{(i)} \Bigg) \begin{bmatrix} X_- \\ U_- \end{bmatrix}^\perp = 0. \label{eq:nullmw}
\end{align}
Rearranging \eqref{eq:nullmw} results in 
\begin{align*}
    \underbrace{(X_+ - C_{\mathcal{M}_w}  ) \begin{bmatrix} X_- \\ U_- \end{bmatrix}^\perp}_{B_{\mathcal{N}_w}} = \sum_{i=1}^{\gamma_{\mathcal{Z}_w} T} \beta^{(i)} \underbrace{G_{\mathcal{M}_w}^{(i)}  \begin{bmatrix} X_- \\ U_- \end{bmatrix}^\perp}_{A^{(i)}_{\mathcal{N}_w}}.
\end{align*}
\end{proof}

\begin{remark}
Algorithm~\ref{alg:LTIConstrainedReachability} provides a less conservative description of the data-driven reachable set compared to Algorithm~\ref{alg:LTIreach} by utilizing a less conservative description of the set of systems consistent with the data. To be more precise, $\mathcal{N}_\Sigma$ in \eqref{eq:N_Sigma} with \eqref{eq:N_w} is an equivalent description of all systems consistent with the data and the noise bound (compare \cite[Lemma 8]{conf:dissipativity1}). However, applying the reachability analysis in line~\ref{ln:algconstRbar} of Algorithm~\ref{alg:LTIConstrainedReachability} requires multiplying constrained matrix zonotopes by zonotopes and constrained zonotopes. For this multiplication, we introduced a guaranteed over-approximation in Proposition~\ref{prop:cmzconstzonotope}, which hence introduces conservatism into the proposed reachability analysis approach. 
\end{remark}

Note that initial zonotope $\mathcal{X}_0$ captures all the uncertainty in the initial state. Next, we provide a general framework for incorporating side information about the unknown model. 
\subsection{Linear Systems with Side Information}\label{sec:sideinfo}

Consider a scenario in which we have prior side information about the unknown model from the physics of the problem or any other source. It would be beneficial to make use of this side information to have less conservative reachable sets. In the following, we propose a framework to incorporate side information about the unknown model, like decoupled dynamics, partial model knowledge, or prior bounds on entries in the system matrices. More specifically, we consider any side information that can be formulated as 
\begin{align}
    | \bar{Q} \begin{bmatrix}
    A_{\text{tr}} & B_{\text{tr}}
    \end{bmatrix} - \bar{Y} | \leq \bar{R}, \label{eq:sideinfo}
\end{align}
where $\bar{Q}\,{\in}\, \mathbb{R}^{n_s \times n_x}$, $\bar{Y}\,{\in}\, \mathbb{R}^{n_s \times (n_x+n_u)}$, and ${\bar{R}\,{\in}\, \mathbb{R}^{n_s \times (n_x+n_u)}}$ are matrices defining the side information which is known to hold for the true system matrices $\begin{bmatrix}
    A_{\text{tr}} & B_{\text{tr}}
    \end{bmatrix}$. Here, the operators $|\cdot|$ and $\leq$ are element-wise operators. To incorporate such side information into the reachability analysis, we utilize once again the newly introduced set of constrained matrix zonotopes. We introduce a reachability analysis in Algorithm~\ref{alg:LTISideInfoReachability} on the basis of the set of system matrices $\begin{bmatrix}A & B\end{bmatrix}$ that is consistent with the data (including the less conservative noise handling in $\mathcal{N}_w$) as well as the a priori known side information in \eqref{eq:sideinfo}. We denote the reachable set computed based on the side information by $\bar{\mathcal{R}}^{\text{s}}_{k}$. Algorithm~\ref{alg:LTISideInfoReachability} summarizes the required computation to incorporate the side information. After setting $\bar{\mathcal{R}}^{\text{s}}_{0} =\mathcal{X}_{0}$ in line~\ref{ln:algSideRs0}, we compute the exact noise description $\mathcal{N}_w$ and exact set of models $\mathcal{N}_\Sigma$ consistent with the noisy data in lines 2:6 similar to Algorithm~\ref{alg:LTIConstrainedReachability}. Next, we compute the set of models $\mathcal{N}_{\text{s}}$ consistent with the side of the information in line~\ref{ln:algSideGNs1} to line~\ref{ln:algSideNs}. Finally, we compute the recursion of the reachable sets in line~\ref{ln:algSideRs}. The following theorem proves that $\bar{\mathcal{R}}^{\text{s}}_{k} \supseteq \mathcal{R}_{k}$.

\begin{algorithm}[t]
  \caption{LTI-Side-Info-Reachability}
  \label{alg:LTISideInfoReachability}
  \textbf{Input}: input-state trajectories $D = (U_-,X)$, initial set $\mathcal{X}_{0}$, process noise zonotope $\mathcal{Z}_w$ and matrix zonotope $\mathcal{M}_w$, side information in terms of $\bar{Q}$, $\bar{Y}$, and $\bar{R}$, and input zonotope $\mathcal{U}_k$, $\forall k = 0, \dots,N-1$\\
  \textbf{Output}: reachable sets $\bar{\mathcal{R}}_{k}, \forall k = 1, \dots,N$ 
  \begin{algorithmic}[1]
  \State $\bar{\mathcal{R}}^{\text{s}}_{0} =\mathcal{X}_{0}$ \label{ln:algSideRs0}
  \Statex Equivalent to lines~\ref{ln:algconstANw}:\ref{ln:algconstNsigma} of Algorithm~\ref{alg:LTIConstrainedReachability} \label{ln:algSideSee2}
  \setcounter{ALG@line}{6}
  \State $G_{\mathcal{N}_{\text{s}}}^{(i)} = G_{\mathcal{N}_\Sigma}^{(i)}, \quad \forall i=\{1, \dots, \gamma_{\mathcal{Z}_w} T\}$ \label{ln:algSideGNs1}
  \State $G_{\mathcal{N}_{\text{s}}}^{(i)} = 0, \quad  \forall i=\{\gamma_{\mathcal{Z}_w} T +1, \dots,\gamma_{\mathcal{Z}_w} T +n_s(n_x+n_u)\}$ \label{ln:algSideGNs2}
  \State $\tilde{G}_{\mathcal{N}_{\text{s}}}=\begin{bmatrix}  G_{\mathcal{N}_{\text{s}}}^{(1)}&\dots&G_{\mathcal{N}_{\text{s}}}^{(\gamma_{\mathcal{Z}_w} T + n_s(n_x+n_u))}\end{bmatrix}$ \label{ln:algSidetildeGNs}
  \State  $A_{\mathcal{N}_{\text{s}}}^{(i)} = \begin{bmatrix} \multicolumn{2}{c}{A_{\mathcal{N}_\Sigma}^{(i)}} \\ \bar{Q} G_{\mathcal{N}_\Sigma}^{(i)} & 0 \end{bmatrix}, \quad \forall i=\{1, \dots, \gamma_{\mathcal{Z}_w} T\}$ \label{ln:algSideANs1}
  \State  $A_{\mathcal{N}_{\text{s}}}^{(\gamma_{\mathcal{Z}_w} T+k)} = \begin{bmatrix}  \multicolumn{2}{c}{0}\\ -\bar{R}^0_{(i,j)} & 0 \end{bmatrix},\exists\ k \  \forall\  i=\{1, \dots, n_x\},$\newline $j=\{1, \dots, n_u\},\ \text{such that} \  k = \{1,\dots,n_s(n_x+n_u)\}$ \label{ln:algSideANs2}
  \State $\tilde{A}_{\mathcal{N}_{\text{s}}}=\begin{bmatrix}  A_{\mathcal{N}_{\text{s}}}^{(1)}&\dots&A_{\mathcal{N}_{\text{s}}}^{(\gamma_{\mathcal{Z}_w} T+ n_s(n_x+n_u))}\end{bmatrix}$ \label{ln:algSidetildeANs}
  \State $B_{\mathcal{N}_{\text{s}}} = \begin{bmatrix} \multicolumn{2}{c}{B_{\mathcal{N}_\Sigma}} \\ \bar{Y} - \bar{Q} C_{\mathcal{N}_\Sigma} & 0 \end{bmatrix}$  \label{ln:algSideBNs}
  \State $\mathcal{N}_{\text{s}} = \zono{C_{\mathcal{N}_\Sigma},\tilde{G}_{\mathcal{N}_{\text{s}}},\tilde{A}_{\mathcal{N}_{\text{s}}},B_{\mathcal{N}_{\text{s}}}}$  \label{ln:algSideNs}
  \For{$k = 0:N-1$}
  \State $\bar{\mathcal{R}}^{\text{s}}_{k+1} =\mathcal{N}_{\text{s}} (\bar{\mathcal{R}}^{\text{s}}_{k} \times \mathcal{U}_{k}  ) +  \mathcal{Z}_w$ \label{ln:algSideRs}
  \EndFor
  \end{algorithmic}
\end{algorithm}

\begin{theorem}
\label{th:sideinfo}
Given input-state trajectories $D = (U_-,X)$ of the system in \eqref{eq:sys} such that $\begin{bmatrix} 
    X_-^\t & U_-^\t 
    \end{bmatrix}^\t$ has a full row rank, and side information in the form of \eqref{eq:sideinfo}, then, the reachable set computed in Algorithm~\ref{alg:LTISideInfoReachability} over-approximates the exact reachable set, i.e., $\bar{\mathcal{R}}^{\text{s}}_{k} \supseteq \mathcal{R}_{k}$. 
\end{theorem}

\begin{proof}
For all matrices $\begin{bmatrix} A_{\text{s}} & B_{\text{s}} \end{bmatrix}$ that satisfy the side information \eqref{eq:sideinfo}, there exists a matrix $\bar{D} \in \mathbb{R}^{n_s \times (n_x + n_u)}$ with $(\bar{D})_{i,j} \in [-1,1]$ such that
\begin{align}
     \bar{Q} \begin{bmatrix} A_{\text{s}} & B_{\text{s}} \end{bmatrix} - \bar{Y}  = \sum_{i=1}^{n_s} \sum_{j=1}^{n_x+n_u} \bar{R}^0_{(i,j)}\odot \bar{D}. \label{eq:sideodotD}
\end{align}
Additionally, we know that all system matrices consistent with the data $\begin{bmatrix} A & B \end{bmatrix} \in \mathcal{N}_{\Sigma}$ are bounded by the constrained matrix zonotope $N_\Sigma$, i.e.
\begin{align}
    \begin{bmatrix} A & B \end{bmatrix} = C_{\mathcal{N}_\Sigma} + \sum_{i=1}^{\gamma_{\mathcal{Z}_w} T} \beta^{(i)}_{\mathcal{N}_\Sigma} G_{\mathcal{N}_\Sigma}^{(i)}, \label{eq:FAB}
\end{align}
with 
\begin{align}
    \sum_{i=1}^{\gamma_{\mathcal{Z}_w} T} \beta^{(i)}_{\mathcal{N}_\Sigma} A_{\mathcal{N}_\Sigma}^{(i)} = B_{\mathcal{N}_\Sigma}. \label{eq:FABcont}
\end{align}
Inserting \eqref{eq:FAB} in \eqref{eq:sideodotD} results in 
\begin{align}
    \bar{Y} - \bar{Q} C_{\mathcal{N}_\Sigma} =\bar{Q} \sum_{i=1}^{\gamma_{\mathcal{Z}_w} T} \beta^{(i)}_{\mathcal{N}_\Sigma} G_{\mathcal{N}_\Sigma}^{(i)} -\sum_{i=1}^{n_s} \sum_{j=1}^{n_x+n_u} \bar{R}^0_{(i,j)}\odot \bar{D}. \label{eq:sideodotD_CG}
\end{align}
With $(\bar{D})_{i,j} \in [-1,1]$, we can concatenate $(\bar{D})_{i,j}$ to $\beta_{\mathcal{N}_\Sigma}$ constituting $\beta_{\mathcal{N}_{\text{s}}}$. Then, combining \eqref{eq:FABcont} with the new constraints in \eqref{eq:sideodotD_CG} yields $\tilde{A}_{\mathcal{N}_{\text{s}}}$ and $B_{\mathcal{N}_{\text{s}}}$. We add zero generators to maintain the correct number of generators.
\end{proof}
\begin{remark}
Note that the reachable sets computed in Algorithm~\ref{alg:LTISideInfoReachability} using side information are less conservative than the ones computed in Algorithm~\ref{alg:LTIConstrainedReachability} using constrained matrix zonotopes which in turn are less conservative than the ones computed in Algorithm~\ref{alg:LTIreach} using the matrix zonotope, i.e., $\mathcal{R}_k \subseteq \bar{\mathcal{R}}^{\text{s}}_{k} \subseteq \bar{\mathcal{R}}_{k} \subseteq \hat{\mathcal{R}}_k$, as additional information is included in the form of additional constraints.
\end{remark}

 \begin{remark}
 The multiplication between a zonotope and matrix zonotopes is computed exactly in line~\ref{ln:algLTIRhat} of Algorithm~\ref{alg:LTIreach}. However, the multiplication between a constrained zonotope and a constrained matrix zonotope is over-approximated using Proposition~\ref{prop:cmzconstzonotope} in line~\ref{ln:algconstRbar} of Algorithm~\ref{alg:LTIConstrainedReachability} and line~\ref{ln:algSideRs} of Algorithm~\ref{alg:LTISideInfoReachability}.
 \end{remark}
\begin{remark}
The computational complexity of our proposed algorithms depends on the number of generators of the reachable sets, the number of generators of the input zonotope, and the number of generators of the matrix zonotope. A reduce operator for zonotopes \cite{Girard2005} or constrained zonotopes \cite{conf:const_zono} is usually used to get over-approximated reachable sets with a lower number of generators at each iteration in order to decrease the complexity. It is $\mathcal{O}(n_x(n_x+n_u)\gamma_{\mathcal{M}_\Sigma}(\gamma_{\hat{\mathcal{R}}}+\gamma_{\mathcal{U}_k}))$ for one step of Algorithm~\ref{alg:LTIreach} due to the multiplication in line~\ref{ln:algLTIRhat}. 
\end{remark}
Next, we consider dealing with measurement noise in combination with process noise. 
%

\subsection{Linear Systems with Measurement Noise} 
\label{sec:measnoise}
In the following, we consider measurement noise in addition to process noise, i.e., 
\begin{align}
\begin{split}
    x(k+1) &= A_{\text{tr}} x(k) + B_{\text{tr}} u(k) + w(k),\\
     y(k) &= x(k) + v(k).
    \end{split}
    \label{eq:sys_v}
\end{align}
Besides the input data matrix $U_-$, we collect the noisy state measurements $Y$ in the matrices 
 \begin{align*}
     Y_+ &= \begin{bmatrix} y^{(1)}(1)\dots  y^{(1)}(T_1) \dots  y^{(K)}(1)  \dots  y^{(K)}(T_K) \end{bmatrix}, \nonumber\\
     Y_- &= \begin{bmatrix} y^{(1)}(0) \dots  y^{(1)}(T_1\!-\!1)  \dots  y^{(K)}(0)  \dots  y^{(K)}(T_K\!-\!1) \end{bmatrix}.
 \end{align*}
%
%
Additionally, let $\hat{O} = \hat{V}_+ - A \hat{V}_-$ with
 \begin{align*}
     \hat{V}_+ &= \begin{bmatrix} \hat{v}^{(1)}(1)\dots  \hat{v}^{(1)}(T_1) \dots  \hat{v}^{(K)}(1)  \dots  \hat{v}^{(K)}(T_K) \end{bmatrix}, \nonumber\\
     \hat{V}_- &= \begin{bmatrix} \hat{v}^{(1)}(0) \dots  \hat{v}^{(1)}(T_1\!-\!1)  \dots  \hat{v}^{(K)}(0)  \dots  \hat{v}^{(K)}(T_K\!-\!1) \end{bmatrix},
 \end{align*}
where $\hat{v}^{(i)}(k)$, $k=0,1,\dots,T_i$, denotes again the actual measurement noise sequence on trajectory $i$ that led to the measured input-state trajectories. If we assume knowledge of the bound on $\hat{O}$, the same approach as presented before can be pursued.

\begin{assumption}
\label{as:mmt_noise}
The matrix $\hat{O}$ is bounded by a matrix zonotope $\hat{O} \in \mathcal{M}_o$ which is known.    
\end{assumption}
\begin{proposition}
\label{prop:meas_zono_av}
Given input-state trajectories $(U_-,Y)$ of the system in \eqref{eq:sys_v} such that $\begin{bmatrix} 
    Y_-^\t & U_-^\t 
    \end{bmatrix}^\t$ has a full row rank, then the reachable set
    \begin{align}
\hat{\mathcal{R}}^{\text{m}}_{k+1} = \mathcal{M}_{\tilde{\Sigma}} (\hat{\mathcal{R}}^{\text{m}}_{k} \times \mathcal{U}_k ) +  \mathcal{Z}_w, \quad \hat{\mathcal{R}}^{\text{m}}_{0}=\mathcal{X}_0,
\end{align}
with 
\begin{align}
    \mathcal{M}_{\tilde{\Sigma}} = (Y_+ - \mathcal{M}_o - \mathcal{M}_w) \begin{bmatrix} 
    Y_- \\ U_- 
    \end{bmatrix}^\dagger \label{eq:mtildsigma}
\end{align}
over-approximates the exact reachable set, i.e., $\mathcal{R}_{k} \subseteq  \hat{\mathcal{R}}^{\text{m}}_{k}$. 
\end{proposition}
\begin{proof}
With
\begin{align*}
    Y_+ - (V_+ - A_{\text{tr}} V_-) - W_- = A_{\text{tr}} Y_- + B_{\text{tr}} U_-,
\end{align*}
 the proof follows the proofs of Lemma~\ref{lm:sigmaM} and Theorem~\ref{th:reach_lin} given Assumption~\ref{as:mmt_noise}.
\end{proof}

Next, we utilize the introduced constrained matrix zonotope in Section~\ref{sec:cmz} to find a less conservative set given Assumption~\ref{as:mmt_noise}. 
\begin{proposition}
\label{prop:meas_cmz_av}
Given input-state trajectories $(U_-,Y)$ of the system in \eqref{eq:sys_v} such that $\begin{bmatrix} 
    Y_-^\t & U_-^\t 
    \end{bmatrix}^\t$ has a full row rank, then the reachable set
\begin{align}
\bar{\mathcal{R}}^{\text{m}}_{k+1} = \mathcal{N}_{\tilde{\Sigma}} (\bar{\mathcal{R}}^{\text{m}}_{k} \times \mathcal{U}_{k}  ) +  \mathcal{Z}_w, \quad \bar{\mathcal{R}}^{\text{m}}_{0}=\mathcal{X}_0, \label{eq:barRm}
\end{align}
with
\begin{align}
    \mathcal{N}_{\tilde{\Sigma}} &= \zono{C_{\mathcal{M}_{\tilde{\Sigma}}},\tilde{G}_{\mathcal{M}_{\tilde{\Sigma}}},\tilde{A}_{\mathcal{N}_{\tilde{\Sigma}}},B_{\mathcal{N}_{\tilde{\Sigma}}}}, \label{eq:Ntildesigma} \\
    \tilde{A}_{\mathcal{N}_{\tilde{\Sigma}}}&=\begin{bmatrix}
    A_{\mathcal{N}_{\tilde{\Sigma}}}^{(1)}&\dots&A_{\mathcal{N}_{\tilde{\Sigma}}}^{(\gamma_{\mathcal{M}_o}+\gamma_{\mathcal{M}_w} )}\end{bmatrix},\nonumber\\
    A^{(i)}_{\mathcal{N}_{\tilde{\Sigma}}} &= G^{(i)}_{\mathcal{M}_w}  \begin{bmatrix} Y_- \\ U_- \end{bmatrix}^\perp\!\!, i =\{1,\dots, \gamma_{\mathcal{M}_w} \},\nonumber\\
    A^{(i)}_{\mathcal{N}_{\tilde{\Sigma}}} &=  G^{(i)}_{\mathcal{M}_o} \begin{bmatrix} Y_- \\ U_- \end{bmatrix}^\perp\!\!, i =\{\gamma_{\mathcal{M}_w} +1,\dots, \gamma_{\mathcal{M}_o}+\gamma_{\mathcal{M}_w}  \},\nonumber\\
    B_{\mathcal{N}_{\tilde{\Sigma}}} &= (Y_{+} - C_{\mathcal{M}_w} - C_{\mathcal{M}_o} )\begin{bmatrix} Y_- \\ U_- \end{bmatrix}^\perp,\nonumber
\end{align} 
over-approximates the exact reachable set, i.e., $\bar{\mathcal{R}}^{\text{m}}_{k+1} \supseteq \mathcal{R}_{k+1}$, where $C_{\mathcal{M}_{\tilde{\Sigma}}}$ and $\tilde{G}_{\mathcal{M}_{\tilde{\Sigma}}}$ are defined in \eqref{eq:mtildsigma}.
\end{proposition}
\begin{proof}
Similar to Theorem~\ref{th:reach_lin_cmz}, we have 
\begin{align*}
    (Y_+ - \hat{W}_- - \hat{O}) \begin{bmatrix} Y_- \\ U_- \end{bmatrix}^\perp = 0.
\end{align*}
We do not know $\hat{W}_-$ and $\hat{O}$ but we can bound them by $\hat{W}_- \in \mathcal{M}_w = \zono{C_{\mathcal{M}_w},\tilde{G}_{\mathcal{M}_w}}$ and  $\hat{O} \in \mathcal{M}_o = \zono{C_{\mathcal{M}_o},\tilde{G}_{\mathcal{M}_o}}$. Therefore, we have:
\begin{align}
    \Bigg(Y_+ - C_{\mathcal{M}_w} - C_{\mathcal{M}_o} - \sum_{i=1}^{\gamma_{\mathcal{M}_w} } \beta^{(i)}_{\mathcal{M}_w} G_{\mathcal{M}_w}^{(i)} \nonumber\\
    - \sum_{i=1}^{\gamma_{\mathcal{M}_o} } \beta^{(i)}_{\mathcal{M}_o} G_{\mathcal{M}_o}^{(i)} \Bigg) \begin{bmatrix} Y_- \\ U_- \end{bmatrix}^\perp = 0. \label{eq:nullmw2}
\end{align}
Let $\beta_{\mathcal{N}_{\tilde{\Sigma}}} = \begin{bmatrix}
\beta_{\mathcal{M}_w} &\beta_{\mathcal{M}_o}
\end{bmatrix}$. Thus we rewrite \eqref{eq:nullmw2} as
\begin{align}
    &\underbrace{\Big(Y_+ - C_{\mathcal{M}_w} - C_{\mathcal{M}_o}  \Big) \begin{bmatrix} Y_- \\ U_- \end{bmatrix}^\perp}_{B_{\mathcal{N}_{\tilde{\Sigma}}}} \nonumber\\
    &= \Bigg( \sum_{i=1}^{\gamma_{\mathcal{M}_w} } \beta^{(i)}_{\mathcal{N}_{\tilde{\Sigma}}} G_{\mathcal{M}_w}^{(i)} 
    + \sum_{i=1}^{\gamma_{\mathcal{M}_o} } \beta^{(\gamma_{\mathcal{M}_w} +i)}_{\mathcal{N}_{\tilde{\Sigma}}} G_{\mathcal{M}_o}^{(i)} \Bigg) \begin{bmatrix} Y_- \\ U_- \end{bmatrix}^\perp, \label{eq:nullmw3}
\end{align}
which yields $B_{\mathcal{N}_{\tilde{\Sigma}}}$ and $A_{\mathcal{N}_{\tilde{\Sigma}}}$.
\end{proof}
%


Note that a similar assumption to Assumption~\ref{as:mmt_noise} has been taken in \cite[Asm.~2]{conf:formulas}. However, it might be difficult in practice to find a suitable set $\mathcal{M}_o$ even with a given bound on $v(k)$, $k=0,1,\dots,T$, since $A$ is assumed to be unknown. 
Therefore, we introduce a data-based approximation for the reachable set under the influence of the measurement noise from data. Instead of Assumption~\ref{as:mmt_noise}, we now only consider a bound on $v(k) \in \mathcal{Z}_v$. Similar to the matrix zonotope $\mathcal{M}_w$ of the modeling noise, we have $\mathcal{M}_v=\zono{C_{\mathcal{M}_v},\tilde{G}_{\mathcal{M}_v}}$ where $\hat{V}_+,\hat{V}_- \in \mathcal{M}_v$. 
Algorithm~\ref{alg:LTIMeasReachability} summarizes the proposed approach to deal with measurement noise. The general idea can be described as follows:

\begin{enumerate}
    \item Obtain an approximate model $\tilde{M} $. 
    \item Obtain a zonotope that gives an over-approximation of the model mismatch between the true model and the approximate model $\tilde{M} $, and the term $A_{\text{tr}}V_-$ from data.
\end{enumerate}

\begin{algorithm}[t]
  \caption{LTI-Meas-Reachability}
  \label{alg:LTIMeasReachability}
    \textbf{Input}: input-state trajectories $D = (U_-,Y)$, initial set $\mathcal{X}_{0}$, process noise zonotope $\mathcal{Z}_w$ and matrix zonotope $\mathcal{M}_w$, measurement noise zonotope $\mathcal{Z}_v$ and matrix zonotope $\mathcal{M}_v$, and input zonotope $\mathcal{U}_k$, $\forall k = 0, \dots,N-1$\\
  \textbf{Output}: reachable sets $\hat{\mathcal{R}}_{k}, \forall k = 1, \dots,N$ 
  \begin{algorithmic}[1]
  \State $\tilde{\mathcal{R}}^{\text{m}}_{0} =\mathcal{X}_{0}$
  \State $\tilde{M}  = (Y_+ - C_{\mathcal{M}_v} - C_{\mathcal{M}_w}  )\begin{bmatrix} 
    Y_- \\ U_- 
    \end{bmatrix}^\dagger$ \label{ln:algmeasMdash}
  \State $\overline{AV}  =\max_j \Bigg( {(Y_{+})}_{.,j} - \tilde{M}  \begin{bmatrix}{(Y_{-})}_{.,j}\\ {(U_{-})}_{.,j} \end{bmatrix} \Bigg)$ \label{ln:algmeasAVupper}
\State $\underline{AV}  =\min_j \Bigg( {(Y_{+})}_{.,j} - \tilde{M}  \begin{bmatrix}{(Y_{-})}_{.,j}\\ {(U_{-})}_{.,j} \end{bmatrix} \Bigg)$ \label{ln:algmeasAVlower}
\State $\mathcal{Z}_{AV} = \text{zonotope}(\underline{AV},\overline{AV}) - \mathcal{Z}_w - \mathcal{Z}_v$ \label{ln:algmeasZAV}
  \For{$k = 0:N-1$} 
  \State $\tilde{\mathcal{R}}^{\text{m}}_{k+1} = \tilde{M} \bigg( (\tilde{\mathcal{R}}^{\text{m}}_k + \mathcal{Z}_v) \times \mathcal{U} \bigg) + \mathcal{Z}_{AV} +\mathcal{Z}_w$ \label{ln:algmeasRtilde}
  \EndFor
  \end{algorithmic}
\end{algorithm}

We obtain an approximate model using a least-squares approach as shown in line~\ref{ln:algmeasMdash} of Algorithm~\ref{alg:LTIMeasReachability}. Rewriting \eqref{eq:sys_v} in terms of the available data results in
\begin{align}
    Y_+ - V_+ = (\tilde{M}  + \Delta \tilde{M} ) \begin{bmatrix} Y_- \\ U_- \end{bmatrix} - A_{\text{tr}} V_- + W_-, \label{eq:xpvp}
\end{align}
where $\Delta \tilde{M} $ is the model mismatch, i.e., $\Delta \tilde{M}  =\begin{bmatrix}A_{\text{tr}} & B_{\text{tr}} \end{bmatrix} - \tilde{M} $. Rearranging \eqref{eq:xpvp} 
to have the terms for which we do not have a bound on the left-hand side results in 
\begin{align}
 \Delta \tilde{M}  \begin{bmatrix} Y_- \\ U_- \end{bmatrix} - A_{\text{tr}} V_- =  Y_+ - \tilde{M} \begin{bmatrix} Y_- \\ U_- \end{bmatrix} -  W_- - V_+   . \label{eq:xpvp2}
\end{align}
%
%
We aim to find a zonotope $\mathcal{Z}_{AV}$ such that, $\forall j = 0,...,T-1$,
\begin{align}
 {(Y_{+})}_{.,j} - \tilde{M}  \begin{bmatrix}{(Y_{-})}_{.,j}\\ {(U_{-})}_{.,j} \end{bmatrix} -  {(W_-)}_{.,j} - {(V_+)}_{.,j} \in \mathcal{Z}_{AV}. \label{eq:ZAV}
\end{align}
To do so, we compute $\overline{AV}$ and $\underline{AV}$ in lines~\ref{ln:algmeasAVupper} and \ref{ln:algmeasAVlower}, respectively. Then, $\mathcal{Z}_{AV}$ is computed in line~\ref{ln:algmeasZAV}. 
Given that $ \Delta \tilde{M} \begin{bmatrix} {(Y_-)}_{.,j} \\ {(U_-)}_{.,j} \end{bmatrix}- A_{\text{tr}} {(V_-)}_{.,j}  \in \mathcal{Z}_{AV}$, $\forall~j~=~0,\dots,T-1$, and $X_+ = Y_+ - V_+$, we rewrite \eqref{eq:xpvp} in terms of sets starting from $\tilde{\mathcal{R}}^{\text{m}}_{0}=\mathcal{X}_0$ as shown in line~\ref{ln:algmeasRtilde}.

\begin{remark}
Algorithm~\ref{alg:LTIMeasReachability} provides a practical approach for computing the reachable set from noisy data (including process and measurement noise). In order to guarantee that the resulting reachable set is indeed an over-approximation of the true reachable set, one would need to assume that the data contains the upper and lower bounds on 
$  \Delta \tilde{M}  \begin{bmatrix} Y_- \\ U_- \end{bmatrix} - A_{\text{tr}} V_- $. Mathematically speaking, we would require the availability of data points 
at some indices $i_1$ and $i_2$ for which  
\begin{align}
\Delta \tilde{M}  z_{i_1} - A_{\text{tr}} v_{i_1} \geq \Delta \tilde{M}  z - A_{\text{tr}} v,\,\,\forall v\in \mathcal{Z}_v, z \in \mathcal{F}. \nonumber \\
\Delta \tilde{M}  z_{i_2} - A_{\text{tr}} v_{i_2} \leq \Delta \tilde{M}  z - A_{\text{tr}} v,\,\,\forall v\in \mathcal{Z}_v, z \in \mathcal{F}. \nonumber 
\end{align} 
holds, where $z_i = \begin{bmatrix} (X_-)_{\cdot,i} \\ (U_-)_{\cdot,i} \end{bmatrix} \in D$.
However, even if this condition is not satisfied and no formal guarantees can be provided, the above approach showed its potential correctly over-approximating the reachable sets in numerical examples.
\end{remark}

\section{Data-Driven Reachability for Nonlinear Systems}\label{sec:reachnonlinear}

 We consider two classes of nonlinear systems, namely, polynomial systems and Lipschitz nonlinear systems.

\subsection{Polynomial Systems}\label{sec:poly}
We consider next a polynomial discrete-time control system
\begin{align}
    x(k+1) &= f_p(x(k),u(k))+ w(k), 
    \label{eq:sysnonlin_poly}
\end{align}
where $f_p:\mathbb{R}^{n_x}\times\mathbb{R}^{n_u} \rightarrow \mathbb{R}^{n_x}$ is a polynomial nonlinear function. In the interest of clarity, we will sometimes omit $k$ as the argument of signal variables, however, the dependence on $k$ should be understood implicitly. 
Let $n_{z}=n_x+n_u$ and
\begin{align}
z=\begin{bmatrix} x^\t & u^\t \end{bmatrix}^\t=\begin{bmatrix} z_1^\t & \dots & z_{n_{z}}^\t \end{bmatrix}^\t \in \mathbb{R}^{n_{z}}
\end{align}
with a misuse of the notations. By a polynomial system, we mean that $f_p(z) \in \mathbb{R}[z]^{n_x}$ is a polynomial nonlinearity, where $\mathbb{R}[z]^{n_x}$ is an $n_x$-dimensional vector with entries in $\mathbb{R}[z]$, which is the set of all polynomials in the variables $z_1,\dots,z_{n_{z}}$ of some degree $d>0$ given by
\begin{align*}
 f_p^{(i)}(z) = \sum_{j=1}^{m_i} \theta_j z_1^{\alpha_{j,1}} z_2^{\alpha_{j,2}} \dots z_{n_{z}}^{\alpha_{j,n_{z}}} 
\end{align*}
with $m_i$ the number of terms in $f_p^{(i)}(z)$, $\theta_j \in \mathbb{R}$ the coefficients, and $\alpha_j = \begin{bmatrix} \alpha_{j,1} & \dots & \alpha_{j,n_{z}} \end{bmatrix}^\t\in\mathbb{N}_0^{n_{z}}$ the vectors of exponents with $\sum_{i=1}^{n_{z}} \alpha_{j,i} \leq d$, for every $j\in\{1,\dots,m_i\}$. 

We write $f_p(z)$ as follows (see \cite{conf:polyDissipat})
\begin{align}
    f_p(z) &{=} \Theta_{\text{tr}} \, h(z) 
    \label{eq:pcg}
\end{align}
%
%
%
where $h(z) \in \mathbb{R}[z]^{m_a}$ contains at least all the monomials present in $f_p(z)$ and ${\Theta_{\text{tr}} \in \mathbb{R}^{n_x \times m_a}}$ contains the unknown coefficients of the monomials in $h(z)$. These monomials can be included in $h(z)$ if, for instance, the upper bound on the degree of polynomials in $f_p(z)$ is known. Moreover, if the structure of the polynomial function $f_p(z)$ is known, then $h(z)$ contains all the monomials of $f_p(z)$. Similarly to the definition of $\mathcal{N}_\Sigma$ in \eqref{eq:Nsig}, we denote the set of unknown coefficients consistent with the data including the true coefficients $\Theta_{\text{tr}}$ by $\mathcal{N}_\Sigma^p$. 
From \eqref{eq:pcg}, let
\begin{align*}
\Omega=\begin{bmatrix} h(x(0),u(0)) \, \dots \, h(x(T-1),u(T-1))\end{bmatrix}
\end{align*}
Then, the following result computes a set of coefficients that is consistent with the data and includes the true coefficients~$\Theta_{\text{tr}}$.

\begin{algorithm}[t!]
  \caption{Polynomial-Reachability}
  \label{alg:PolyReachability}
  \textbf{Input}: input-state trajectories $D = (U_-,X)$ of the polynomial system in \eqref{eq:sysnonlin_poly}, initial set $\mathcal{X}_{0}$, process noise zonotope $\mathcal{Z}_w$ and matrix zonotope $\mathcal{M}_w$, and input zonotope $\mathcal{U}_k, \forall k = 0, \dots,N-1$. \\
  \textbf{Output}: reachable sets $\hat{\mathcal{R}}_{k}^p, \forall k = 1, \dots,N$
 \begin{algorithmic}[1]
  \State $\hat{\mathcal{R}}_{0}^p =\mathcal{X}_{0}$ \label{ln:R_0}
  \State $\Omega=\begin{bmatrix} h(x(0),u(0)) \, \dots \, h(x(T-1),u(T-1))\end{bmatrix}$
  \State $\mathcal{M}_{\Sigma}^p = (X_{+} - \mathcal{M}_w) \Omega^\dagger$\label{ln:algMsigma}
  \For{$k = 0:N-1$}
  \State $\hat{\mathcal{R}}_{k+1}^p = \mathcal{M}_{\Sigma}^p \, h(\text{int}(\hat{\mathcal{R}}_k^p),\text{int}(\mathcal{U}_{k})) +  \mathcal{Z}_w$ \label{ln:timeupdate}
  \EndFor
  \end{algorithmic}
\end{algorithm}


\begin{lemma}
\label{lm:sigmaM_p}
Given a matrix $\Omega$ of the polynomial system in \eqref{eq:sysnonlin_poly} with a full row rank, then the matrix zonotope 
\begin{align}
    \mathcal{M}_{\Sigma}^p = (X_{+} - \mathcal{M}_w) \Omega^\dagger \label{eq:Msigma}
\end{align} 
 contains all matrices $\Theta$ that are consistent with the data and the noise bound, i.e., $\mathcal{M}_{\Sigma}^p \supseteq \mathcal{N}_{\Sigma}^p$. 
\end{lemma}
\begin{proof}
The proof is similar to Lemma~\ref{lm:sigmaM}. We have from data and the polynomial system in \eqref{eq:sysnonlin_poly}:
\begin{align}
    X_+ = \Theta_{\text{tr}} \Omega  + W_-, \label{eq:xpczw}
\end{align}
where $W_-$ is the noise in the data. We do not know $W_-$ but we can bound it by $W_- \in \mathcal{M}_w$. Hence, rearranging \eqref{eq:xpczw} results in $\mathcal{M}_{\Sigma}^p$ in \eqref{eq:Msigma} where $\Theta_{\text{tr}} \in \mathcal{M}_{\Sigma}^p$ given that $W_- \in \mathcal{M}_w$.
\end{proof}

After computing the set of coefficients that is consistent with the data, the next open question is how to forward propagate the reachable set. In the linear case in Section~\ref{sec:reachlineardis}, we required a linear map and Minkowski sum operations, which are provided by zonotope properties. For polynomial systems, we need to compute monomials of the reachable sets as shown in \eqref{eq:pcg}, which is not possible using zonotopes. Thus, we propose over-approximating the reachable set, represented by a zonotope, by an interval, as it is possible to compute the monomials of an interval set. 

The algorithm is summarized in Algorithm~\ref{alg:PolyReachability}. We first initialize the reachable set $\hat{\mathcal{R}}_{0}^p$ in line~\ref{ln:R_0}. Then, at each time step $k=0,\dots,N-1$, we convert the reachable set and the input set into intervals by writing 
$\text{int}(\hat{\mathcal{R}}_k^p)$ and $\text{int}(\mathcal{U}_{k})$, respectively. Then, we substitute in the list of monomials $h(\text{int}(\hat{\mathcal{R}}_k^p),\text{int}(\mathcal{U}_{k}))$ using interval arithmetic. Then, in line~\ref{ln:timeupdate}, we propagate ahead the estimated set using the  matrix zonotope $\mathcal{M}_{\Sigma}^p$, interval of all monomials $h(\text{int}(\hat{\mathcal{R}}_k^p),\text{int}(\mathcal{U}_{k}))$, and the noise zonotope $\mathcal{Z}_w$. 
%
%

%
%
%
\begin{theorem}
\label{th:reach_lin_poly}
Given a matrix $\Omega$ with a full row rank of the polynomial system in \eqref{eq:sysnonlin_poly}, then the reachable set computed in Algorithm~\ref{alg:PolyReachability} over-approximates the exact reachable set, i.e., $\hat{\mathcal{R}}_{k+1}^p \supseteq \mathcal{R}_{k+1}$.
\end{theorem}
%
%
\begin{proof} Given that $\Theta_{\text{tr}} \in \mathcal{M}_{\Sigma}^p$, $\hat{\mathcal{R}}^p_{k} \subseteq \text{int}(\hat{\mathcal{R}}^p_{k})$, $\mathcal{U}_{k} \subseteq \text{int}(\mathcal{U}_{k})$, and both $\mathcal{R}_k$ and $\hat{\mathcal{R}}_k^p$ start from the same initial set $\mathcal{X}_0$, it holds that $\mathcal{R}_{k+1} \subseteq \hat{\mathcal{R}}_{k+1}^p$.
\end{proof}
%
%
%
\begin{remark}
The condition in Lemma~\ref{lm:sigmaM_p} of requiring $\Omega$ with a full row rank implies that there exists a right-inverse of the matrix $\Omega$. This condition can be easily checked given the data. 
\end{remark}

\begin{remark}
Similar to LTI systems, we can utilize constrained matrix zonotopes to obtain less conservative reachable sets, denoted by $\bar{\mathcal{R}}_{k}^p$, using the improved description of the noise matrix zonotope $\mathcal{N}_w$ and propagating forward using interval arithmetic. Furthermore, we can also include side information
\begin{align*}
    | \bar{Q}^p \Theta_{\text{tr}} - \bar{Y}^p | \leq \bar{R}^p, 
\end{align*}
where $\bar{Q}^p \in \mathbb{R}^{n_s \times n_x}$, $\bar{Y}^p \in \mathbb{R}^{n_s \times m_a}$, and $\bar{R}^p \in \mathbb{R}^{n_s \times m_a}$ are matrices defining the side information which is known to hold for the true system matrix $\Theta_{\text{tr}} \in \mathbb{R}^{n_x \times m_a}$. The reachable sets, denoted by $\bar{\mathcal{R}}^{\text{s},p}_{k}$, while taking the side information into account can be computed similar to LTI systems. This will be evaluated in the evaluation section.
\end{remark}

\subsection{Lipschitz Nonlinear Systems} \label{sec:libs}
We consider a discrete-time Lipschitz nonlinear control system
\begin{align}
    x(k+1) &= f(x(k),u(k))+ w(k). 
    \label{eq:sysnonlin}
\end{align}
We assume in this subsection $f$ to be twice differentiable. A local linearization of \eqref{eq:sysnonlin} is performed by a Taylor series expansion around the linearization point $z^\star=\begin{bmatrix}x^\star \\u^\star \end{bmatrix}$:
\begin{align*}
f(z) =& f(z^\star) + \frac{\partial f(z)}{\partial z}\Big|_{z=z^\star} (z - z^\star)+ \dots 
\end{align*}
The infinite Taylor series \cite{conf:taylor} can be represented by a first-order Taylor series and a Lagrange remainder term $L(z)$ \cite[p.65]{conf:thesisalthoff}, that depends on the model, as follows.
\begin{align}
f(z) = f(z^\star) + \frac{\partial f(z)}{\partial z}\Big|_{z=z^\star} (z - z^\star) 
+ L(z).
\label{eq:linfL}
\end{align}
Since the model is assumed to be unknown, we aim to over-approximate $L(z)$ from data. We rewrite \eqref{eq:linfL} as follows:

\begin{align}
f(x,u) =& f(x^\star,u^\star) + \underbrace{\frac{\partial f(x,u)}{\partial x}\Big|_{x=x^\star,u=u^\star}}_{\tilde{A}} (x - x^\star) \nonumber\\
&+ \underbrace{\frac{\partial f(x,u)}{\partial u}\Big|_{x=x^\star,u=u^\star}}_{\tilde{B}} (u - u^\star) + L(x,u)\nonumber,
\end{align}
i.e.,
\begin{align}
f(x,u) = \begin{bmatrix}f(x^\star,u^\star) & \tilde{A} & \tilde{B}\end{bmatrix} \begin{bmatrix}1\\x-x^\star\\ u-u^\star\end{bmatrix} +L(x,u).
\label{eq:fz_incl}
\end{align}


Algorithm~\ref{alg:LipReachability} shows the proposed approach. We conduct data-driven reachability analysis for nonlinear systems by the following two steps: 
\begin{enumerate}
    \item Obtain an approximate linearized model from the noisy data.
    \item Obtain a zonotope that over-approximates the modeling mismatch together with the Lagrange remainder $L(z)$ for the chosen system. 
\end{enumerate}

\begin{algorithm}[t]
  \caption{Lipschitz-Reachability}
  \label{alg:LipReachability}
  \textbf{Input}: input-state trajectories $D = (U_-,X)$, initial set $\mathcal{X}_{0}$, process noise zonotope $\mathcal{Z}_w$ and matrix zonotope $\mathcal{M}_w$, Lipschitz constant $L^\star$, covering radius $\delta$, and input zonotope $\mathcal{U}_k$, $\forall k = 0, \dots,N-1$\\
  \textbf{Output}: reachable sets $\mathcal{R}^\prime_{k}, \forall k = 1, \dots,N$ 
  \begin{algorithmic}[1]
  \State $\mathcal{R}^\prime_{0} =\mathcal{X}_{0}$
\State $\mathcal{Z}_\epsilon = \zono{0,\textup{diag}({L^\star}^{(1)} \delta/2,\dots,{L^\star}^{(n_x)} \delta/2)}$ \label{ln:alglipZeps}
  \For{$k = 0:N-1$}
    \State $M^\prime = (X_+ - C_{\mathcal{M}_w}) \begin{bmatrix} 
    1_{1 \times T}\\ X_{-}-1 \otimes x^\star(k) \\ U_{-}-1 \otimes u^\star(k) 
    \end{bmatrix}^\dagger$ \label{ln:alglipMtilde}
  \State $\overline{l}  =\max_j \Bigg( {(X_{+})}_{.,j} - M^\prime \begin{bmatrix}1\\ 
{(X_{-})}_{.,j} -  x^\star(k)\\ {(U_{-})}_{.,j} - u^\star(k)\end{bmatrix} \Bigg)$ \label{ln:algliplupper}
\State $\underline{l}  =\min_j \Bigg( {(X_{+})}_{.,j} - M^\prime \begin{bmatrix}1\\ 
{(X_{-})}_{.,j} -  x^\star(k)\\ {(U_{-})}_{.,j} - u^\star(k)\end{bmatrix} \Bigg)$ \label{ln:alglipllower}
\State $\mathcal{Z}_{L} = \text{zonotope}(\underline{l},\overline{l}) - \mathcal{Z}_w$ \label{ln:alglipZl}
  \State $\mathcal{R}^\prime_{k+1} {=} M^\prime \Big(1 \times( \mathcal{R}^\prime_{k} -x^\star )\times (\mathcal{U}_k - u^\star )\Big) +  \mathcal{Z}_w +  \mathcal{Z}_L + \mathcal{Z}_\epsilon$\label{ln:alglipRprime}
  \EndFor
  \end{algorithmic}
\end{algorithm}

To obtain an approximate linearized model, we apply a least-squares approach. Without additional knowledge on $L(z)$ and $w(k) \in \mathcal{Z}_w$ (or $W_- \in \mathcal{M}_w = \langle C_{\mathcal{M}_w}, \tilde{G}_{\mathcal{M}_w}  \rangle$), a best guess in terms of a least-square approach is $M^\prime$ in line~\ref{ln:alglipMtilde} of Algorithm~\ref{alg:LipReachability}. To over-approximate the remainder term $L(z)$ from data, we need to assume that $f$ is Lipschitz continuous for all $z$ in the reachable set $\mathcal{F}$ as defined in \eqref{eq:F}.



\begin{assumption}
It holds that $f: \mathcal{F} \rightarrow \mathbb{R}^{n_x}$ is Lipschitz continuous, i.e., that there is some $L^\star \geq 0$ such that 
$\| f(z) - f(z^{\prime}) \|_2 \leq L^\star \| z - z^{\prime}\|_2$
holds for all $z, z^{\prime} \in \mathcal{F}$.
\label{as:lipschitz}
\end{assumption}

For data-driven methods of nonlinear systems, Lipschitz continuity is a common assumption (e.g. \cite{conf:montenbruckLipschitz,conf:novaraLipschitz}). 
%
%
By compactness of $\mathcal{U}_k$, $\mathcal{R}_k$, $k=0,\dots,N$, also $\mathcal{F}$ is compact. Therefore, the data points $D = (U_-,X)$ are relatively dense in $\mathcal{F}$ such that for any $z \in \mathcal{F}$  there exists a $z_i = \begin{bmatrix} (X_-)_{\cdot,i} \\ (U_-)_{\cdot,i} \end{bmatrix} \in D$ such that $\|z - z_i \| \leq \delta$. The quantity $\delta$ is sometimes referred to as the covering radius or the dispersion. The following theorem proves the over-approximation of the reachable sets $\mathcal{R}^\prime_{k}$ out of Algorithm~\ref{alg:LipReachability} for the exact reachable sets $\mathcal{R}_k$.




\begin{theorem}
\label{th:reachdisnonlin}
Given data $D = (U_-,X)$ from a system in \eqref{eq:sysnonlin}, then the reachable set computed in Algorithm~\ref{alg:LipReachability} over-approximates the exact reachable set, i.e., $\mathcal{R}_{k} \subseteq \mathcal{R}^\prime_{k}$.
\end{theorem}

\begin{proof}
We know from~\eqref{eq:fz_incl} that
\begin{align*}
    f(z) = (M^\prime + \Delta M^\prime) \begin{bmatrix} 1 \\ z - z^\star \end{bmatrix} + L(z),
\end{align*}
where $\Delta M^\prime$ captures the model mismatch defined by $\Delta M^\prime = \begin{bmatrix} f(z^\star) & \tilde{A} & \tilde{B} \end{bmatrix} - M^\prime$.
Hence, we need to show that $\mathcal{Z}_L + \mathcal{Z}_\epsilon$ over-approximates the modeling mismatch and the term $L(z)$, i.e.,
\begin{align*}
    \Delta M^\prime \begin{bmatrix} 1 \\ z - z^\star \end{bmatrix} + L(z) \in \mathcal{Z}_L + \mathcal{Z}_\epsilon
\end{align*}
for all $z \in \mathcal{F}$. We start by proving for the available data  $z_i \in D = (U_-,X)$ then generalize to $z \in \mathcal{F}$. We know that for all $z_i \in D = (U_-,X)$ and $(W_-)_{\cdot, i} \in \mathcal{Z}_w$, it holds that
\begin{align*}
    (X_+)_{\cdot,i} - (W_-)_{\cdot, i} = (M^\prime + \Delta M^\prime) \begin{bmatrix} 1 \\ z_i - z^\star \end{bmatrix} + L(z_i),
\end{align*}
which implies
\begin{align}
     (X_+)_{\cdot,i}  - M^\prime \begin{bmatrix} 1 \\ z_i - z^\star \end{bmatrix} \in
     \Delta M^\prime \begin{bmatrix} 1 \\ z_i - z^\star \end{bmatrix} + L(z_i) + \mathcal{Z}_w.
    \label{eq:deltamindata}
\end{align}
Next, we aim to find one zonotope $\mathcal{Z}_L$ that over-approximates $\Delta M^\prime \begin{bmatrix} 1 \\ z_i - z^\star \end{bmatrix} + L(z_i)$ for all the data points, i.e.,  $\forall z_i \in D$
\begin{align*}
 (X_+)_{\cdot,i} - M^\prime \begin{bmatrix} 1 \\ z_i - z^\star \end{bmatrix}
 \in \mathcal{Z}_L + \mathcal{Z}_w.    
\end{align*}
This can be done by finding the upper bound ($\overline{l}$ in line~\ref{ln:algliplupper}) and lower bound ($\underline{l}$ in line~\ref{ln:alglipllower}) from data and hence $\mathcal{Z}_L$ in line~\ref{ln:alglipZl}. Thus, we can over-approximate the model mismatch and the nonlinearity term for all data points $z_i \in D = (U_-,X)$, $i=0,1,\dots,T$, by
\begin{align*}
    f(z_i) \in M^\prime \begin{bmatrix} 1 \\ z_i - z^\star  \end{bmatrix} + \mathcal{Z}_L.
\end{align*}
Given the covering radius $\delta$ of our system together with Assumption~\ref{as:lipschitz}, we know that for every $z \in \mathcal{F}$, there exists a $z_i \in D = (U_-,X)$ such that
    $\| f(z) - f(z_i) \| \leq L^\star \| z - z_i \| \leq L^\star \delta$.
This yields 
\begin{align*}
    f(z) \in M^\prime \begin{bmatrix} 1 \\ z_i - z^\star \end{bmatrix} +  \mathcal{Z}_L + \mathcal{Z}_\epsilon,
\end{align*}
with $\mathcal{Z}_\epsilon = \zono{0,\text{diag}(L^\star \delta/2,\dots,L^\star \delta/2)}$.
\end{proof}
For an infinite amount of data, i.e., $\delta \rightarrow 0$, we can see that $\mathcal{Z}_\epsilon \rightarrow 0$, i.e., the formal $\mathcal{Z}_L$ then fully captures the modeling mismatch and the Lagrange reminder. 
Also, we would like to note that our approach works with any type of right inverse. 

\begin{remark}
\label{rem:approx}
Note that determining $L^\star$ as well as computing $\delta$ is non-trivial in practice. If we assume that the data is evenly spread out in the compact input set of $f$, then the following can be a good approximation of the upper bound on $L^\star$ and $\delta$ for each dimension $o$:
\begin{align*}
     \hat{L}^{\star^{(o)}} &= \max_{z_i, z_j \in D, i\neq j}  \frac{\| f^{(o)}(z_i) - f^{(o)}(z_j) \|}{\| z_i - z_j \|} \\
    \hat{\delta} &= \max_{z_i \in D} \min_{z_j \in D, j \neq i} \| z_i - z_j \|.
\end{align*}
Computing the Lipschitz constant for each dimension decreases the conservatism, especially when the data has a different scale for each dimension. Other methods to calculate the Lipschitz constant $L^\star$ can be found in~\cite{conf:montenbruckLipschitz,conf:novaraLipschitz}, and a sampling strategy to obtain a specific $\delta$ is introduced in~\cite{conf:montenbruckLipschitz}.

Furthermore, note that from the proof of the above theorem, we see that for every reachability step, $k=1, \dots, N$, local information on $L^\star$ and $\delta$ in the set $\mathcal{R}_k \times \mathcal{U}_k $ can be used, if available, to reduce conservatism. 
\end{remark}
\begin{remark}
The Lipschitz constant $L^\star$ and the covering radius $\delta$ are required to hold within the set $\mathcal{F}$. In practice, however, $\mathcal{F}$ is not known a priori. However, any over-approximation on $\mathcal{F}$ would also be sufficient in this sense. Taking any over-approximation would lead to the same guarantees but might result in larger required data sets or a more conservative reachability analysis. This requirement also makes sense on an intuitive level: We need data from all regions of significance for the reachability analysis in the general case of nonlinear systems.
\end{remark}

\begin{remark}
We choose the linearization points as the center of the current input zonotope $\mathcal{U}_k$ and state zonotope $\mathcal{R}^\prime_k$, and we repeat the linearization at each time step $k$. In model-based reachability analysis, the optimal linearization point is the center of the current state and input zonotopes as proved in \cite[Corollary 3.2]{conf:thesisalthoff}, which then minimizes the set of Lagrange remainders. Therefore, choosing the center of the current input and state zonotopes as linearization points is a natural choice, but the theoretical results are independent of this choice.
\end{remark}

\section{Evaluation}\label{sec:eval}

In this section, we apply the computational approaches to over-approximate the reachable exact sets from data. Firstly, we consider simulative data from a discrete-time LTI system, a polynomial system, and a nonlinear discrete-time system. Then, we collected real-world data from an autonomous car and performed the respective experiments. Throughout the section, we misuse the notation $\mathcal{R}_k$ to denote also the model-based reachable set. 
%
\begin{figure*}[!htbp]
\vspace{-0.05cm}
    \centering
    \begin{subfigure}[h]{0.32\textwidth}
     \centering
        \includegraphics[scale=0.26]{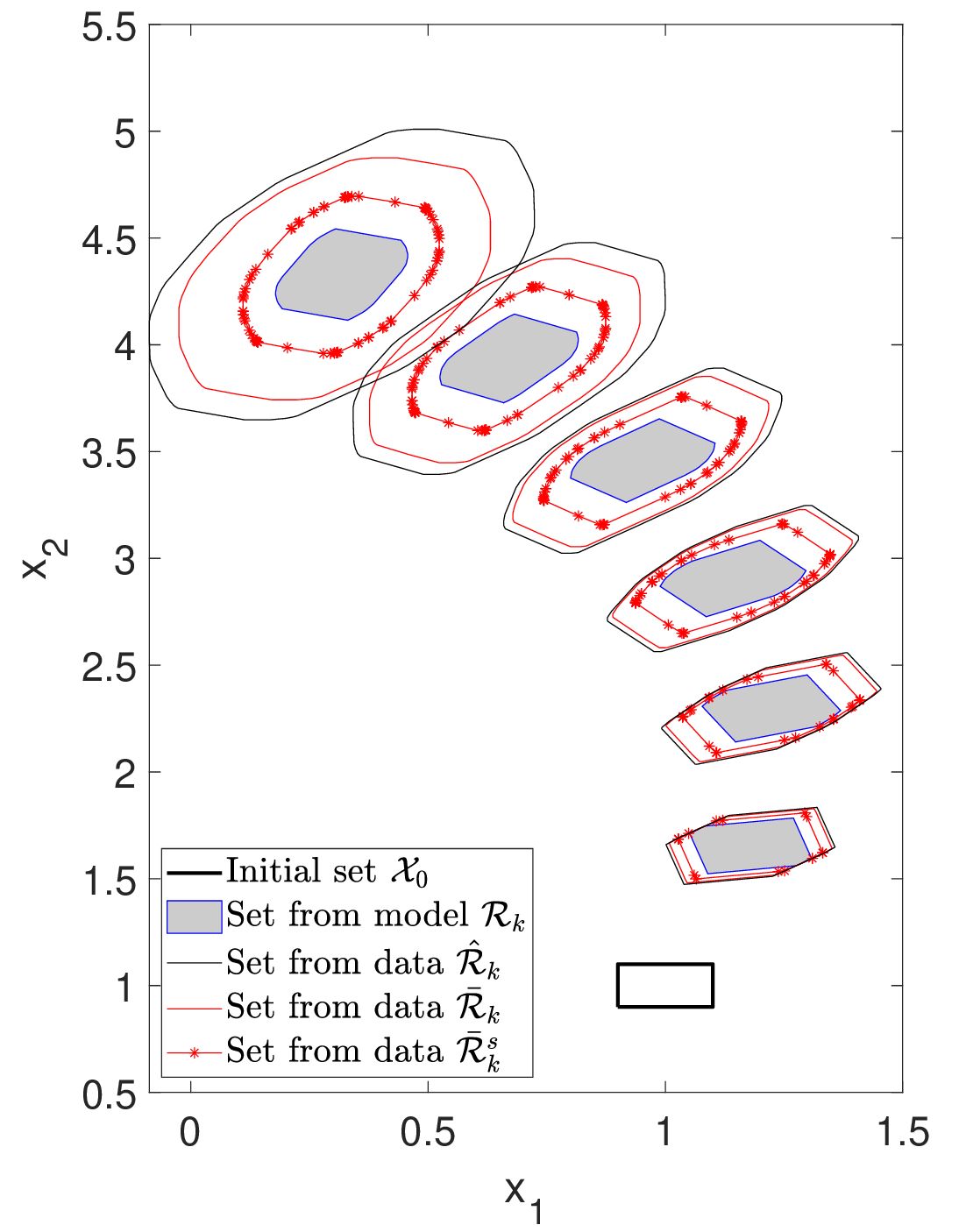}
        \caption{}
        \label{fig:x1x2sidew0-005-In10S3-red200}
    \end{subfigure}
    \begin{subfigure}[h]{0.32\textwidth}
     \centering
        \includegraphics[scale=0.26]{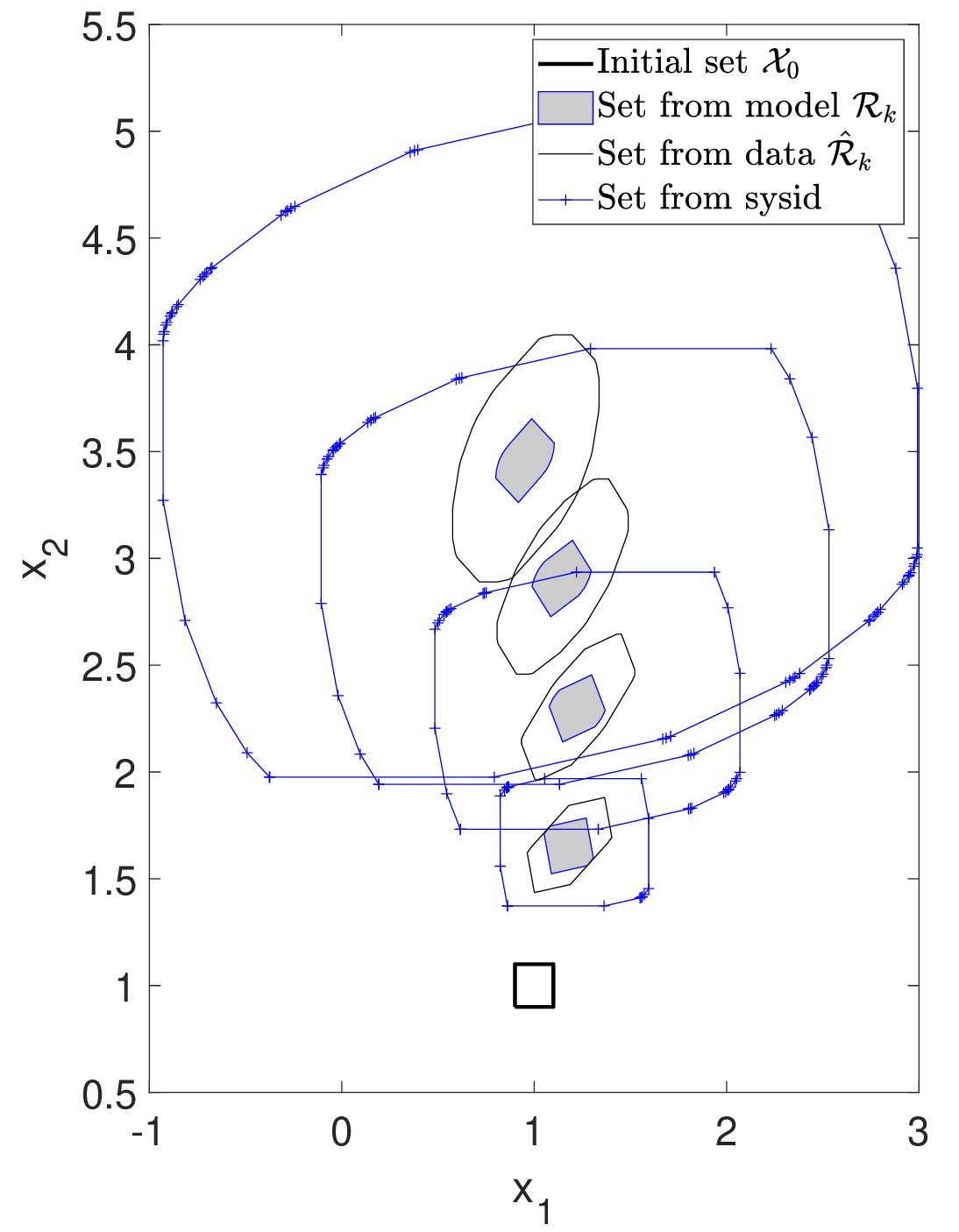}
        \caption{}
        \label{fig:x1x2_sysid}
    \end{subfigure}
    \begin{subfigure}[h]{0.32\textwidth}
     \centering
        \includegraphics[scale=0.26]{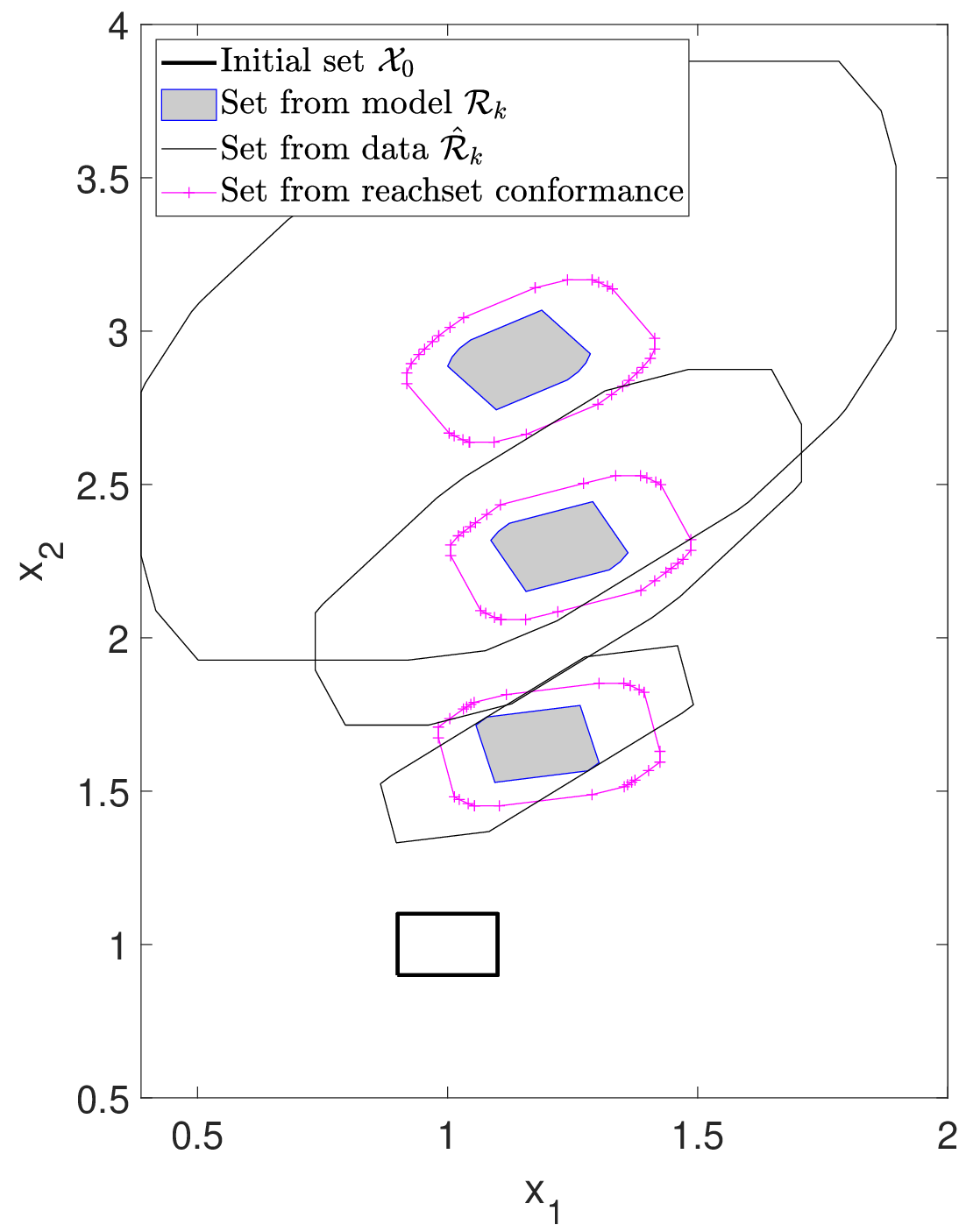}
        \caption{}
        \label{fig:reachconfx12}
    \end{subfigure}
\caption{The projection of the reachable sets of the LTI system in \eqref{eq:sysexamplediscrete} computed via Algorithm~\ref{alg:LTIreach} ($\hat{\mathcal{R}}_k$), Algorithm~\ref{alg:LTIConstrainedReachability} ($\bar{\mathcal{R}}_k$), and Algorithm~\ref{alg:LTISideInfoReachability} ($\bar{\mathcal{R}}_k^{\text{s}}$) from noisy input-state data is presented in (a). We compute in (b) and (c) the reachable sets via standard subspace system identification method (\textit{N4SID}) including $2\sigma$ uncertainty bound in the analysis and via synthesizing a reachset conformant model \cite{Kochdumper2020}, respectively. For better comparison, we also add $\hat{\mathcal{R}}_k$ from (a) to (b) and (c) (as the size required differently scaled axis).}
    \label{fig:contreach}
     \vspace{-6mm}
\end{figure*}

\subsection{LTI Systems}
To demonstrate the usefulness of the presented approach, we consider the reachability analysis of a five-dimensional system which is a discretization of the system used in \cite[p.39]{conf:thesisalthoff} with sampling time $0.05$ sec. 
The system has the following model.
\begin{align}
\begin{split}
A_{\text{tr}}&=\begin{bmatrix}
    0.9323 &  -0.1890   &      0   &      0   &      0 \\
    0.1890 &   0.9323  &       0  &       0   &      0 \\
         0 &        0  &  0.8596  &   0.0430  &        0 \\
         0 &         0   & -0.0430    & 0.8596      &    0 \\
         0 &         0  &        0    &      0   &  0.9048
\end{bmatrix}, \\
B_{\text{tr}}&=\begin{bmatrix} 
    0.0436&
    0.0533&
    0.0475&
    0.0453&
    0.0476
    \end{bmatrix}^\t.
\end{split}
\label{eq:sysexamplediscrete}
\end{align}
The initial set is chosen to be $\mathcal{X}_0=\zono{1,0.1 I}$ where $1$ and $I$ are vectors of one and the identity matrix, respectively. The input set is $\mathcal{U}_k=\zono{10,0.25}$. We consider computing the reachable set when there is a random noise sampled from the zonotope  $\mathcal{Z}_w=\zono{0,\begin{bmatrix}0.005 \, \dots \, 0.005\end{bmatrix}^T}$. Three trajectories of length 10 ($T=30$) are used as input data $D$. We present in Fig.~\ref{fig:x1x2sidew0-005-In10S3-red200} the projections of the following reachable sets on the first two dimensions:


\begin{itemize}
    \item The true model based reachable sets $\mathcal{R}_k$. 
    \item The reachable sets $\hat{\mathcal{R}}_k$ computed via Algorithm~\ref{alg:LTIreach} using matrix zonotopes.
    \item The reachable sets $\bar{\mathcal{R}}_k$ computed via Algorithm~\ref{alg:LTIConstrainedReachability} using constrained matrix zonotopes.
    \item The reachable sets $\bar{\mathcal{R}}_k^{\text{s}}$ utilizing the states decoupling as a side information computed via Algorithm~\ref{alg:LTISideInfoReachability}. The used parameters are $\bar{Q} = I$, $\bar{Y}=0$ and 
    \begin{align*}
       \bar{R}&=\begin{bmatrix}
    1 &  1   &      0.001   &      0.001   &      0.001  &1\\
    1 &   1  &       0.001  &       0.001   &      0.001 &1\\
         0.001 &        0.001  &  1  &   1  &        0.001 &1\\
         0.001 &         0.001   & 1    & 1      &   0.001 &1\\
         0.001 &         0.001  &        0.001    &      0.001   &  1&1
\end{bmatrix}. 
    \end{align*}
\end{itemize}

We compare the different data-driven reachability results to the true reachable set computed via model-based reachability analysis given the exact underlying model. Consistent with the theoretical analysis and guarantees derived in this work, the data-driven reachability results correctly over-approximate the true reachable sets at all times. 

We measured the execution time of the proposed algorithms in comparison to the model-based algorithm, which is done by the linear map and Minkowski sum operations \cite[p.17]{conf:thesisalthoff}. The experiments were done on a processor 11$^{th}$ Generation Intel(R) Core(TM) i7-1185G7 with 16.0 GB RAM. Table~\ref{tab:execTimeLinear} shows the execution time in minutes. Analysing Table~\ref{tab:execTimeLinear} and Fig.~\ref{fig:x1x2sidew0-005-In10S3-red200} shows a trade-off between the size of the reachable sets and the execution time. 

\begin{table}[t!]
\caption{Execution time in minutes for reachability analysis of the LTI system.}
\label{tab:execTimeLinear}
\centering
\normalsize
\begin{tabular}{c  c  }
\toprule
 Algorithm & Execution time \\
\midrule
Model-based reachability & $1.178 \times 10^{-04}$ \\
Algorithm~\ref{alg:LTIreach} & $1.366\times 10^{-04}$\\
Algorithm~\ref{alg:LTIConstrainedReachability}& $0.208$\\
Algorithm~\ref{alg:LTISideInfoReachability}& $0.325$\\
\bottomrule
\end{tabular}
\vspace{-5mm}
\end{table}

Furthermore, we compare the data-driven reachability results with one standard system identification method and apply state-of-the-art reachability analysis with the identified model. More specifically, we apply the \textit{N4SID} subspace identification algorithm \cite{conf:sysid} to the noisy data, and show the zonotope describing the resulting reachable set corresponding to the $2\sigma$ uncertainty bound in Fig.~\ref{fig:x1x2_sysid}. In comparison, the \textit{N4SID} reachable sets are quite conservative. We also compare our algorithms with the reachset conformance technique~\cite{Kochdumper2020} in Fig.~\ref{fig:reachconfx12}. We added a high amount of noise in the data by assuming the noise zonotope to be $\mathcal{Z}_w=\zono{0,\begin{bmatrix} 0.03 \, \dots \, 0.03 \end{bmatrix}^T}$ to quantify the main differences between the two approaches in Fig.~\ref{fig:reachconfx12}. Our approach comes with robust guarantees; however, it is conservative, especially with a high amount of noise. The reach conformance can't guarantee state inclusion in the computed set for the unseen measurements, which is guaranteed in our approach.  


Next, we consider the same problem setup but with additional measurement noise in the data as described in \eqref{eq:sys_v}, where $\mathcal{Z}_w=\zono{0,\begin{bmatrix}0.005 \, \dots \, 0.005\end{bmatrix}^T}$, $\mathcal{Z}_v=\zono{0,\begin{bmatrix}0.002 \, \dots \, 0.002 \end{bmatrix}^T}$, and $\mathcal{M}_o = \mathcal{M}_v - A_{\text{tr}} \mathcal{M}_v$ as an assumed a priori known over-approximation. The results of applying the approaches introduced in Section~\ref{sec:reachlineardis} can be seen in Fig.~\ref{fig:x1x2meas_v0_002w0-005-In10S2-red390} which shows the following sets



\begin{figure*}[!htbp]
\vspace{-0.05cm}
    \centering
    \begin{subfigure}[h]{0.32\textwidth}
     \centering
        \includegraphics[scale=0.26]{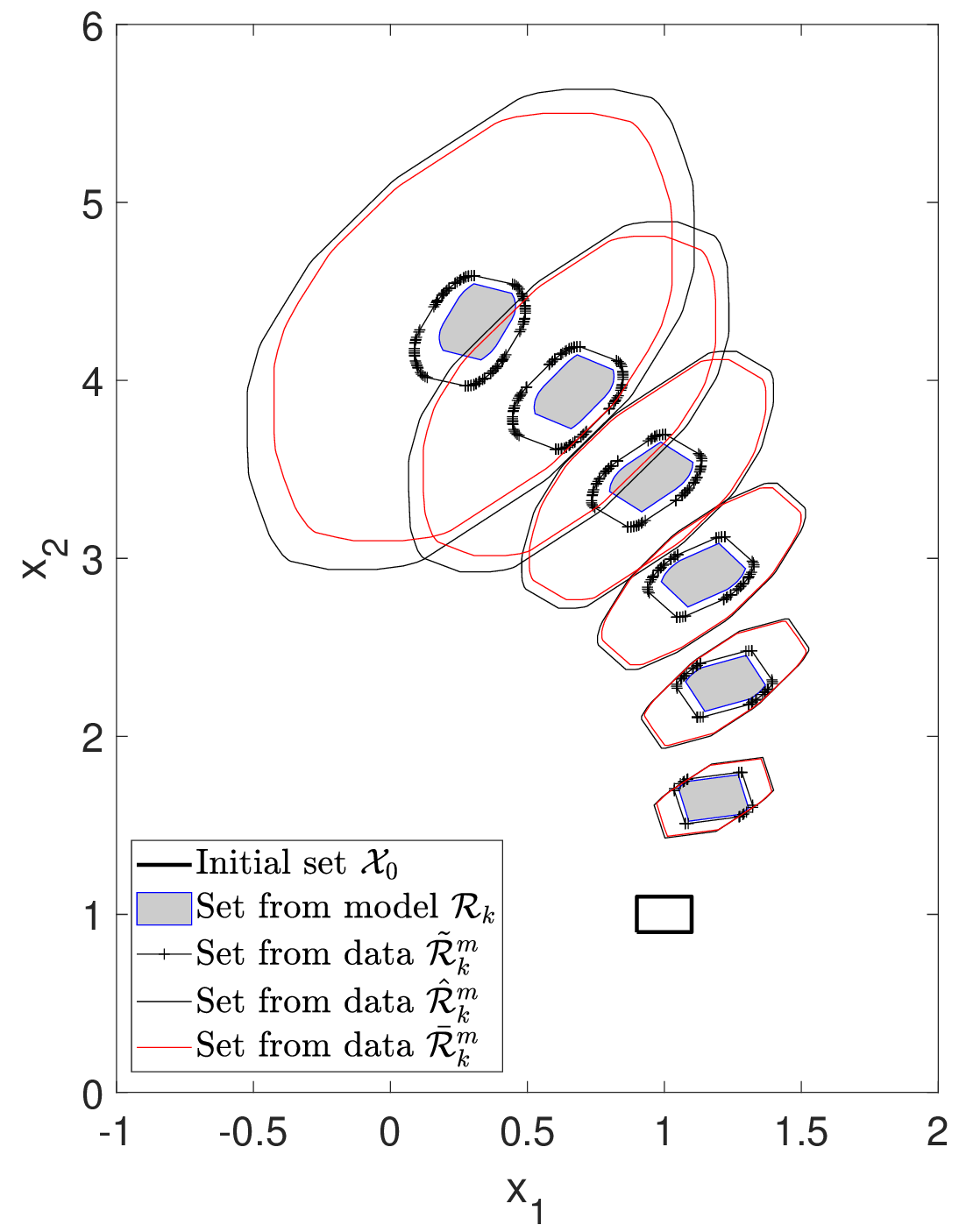}
        \caption{}
        \label{fig:x1x2meas_v0_002w0-005-In10S2-red390}
    \end{subfigure}
    \begin{subfigure}[h]{0.32\textwidth}
     \centering
        \includegraphics[scale=0.26]{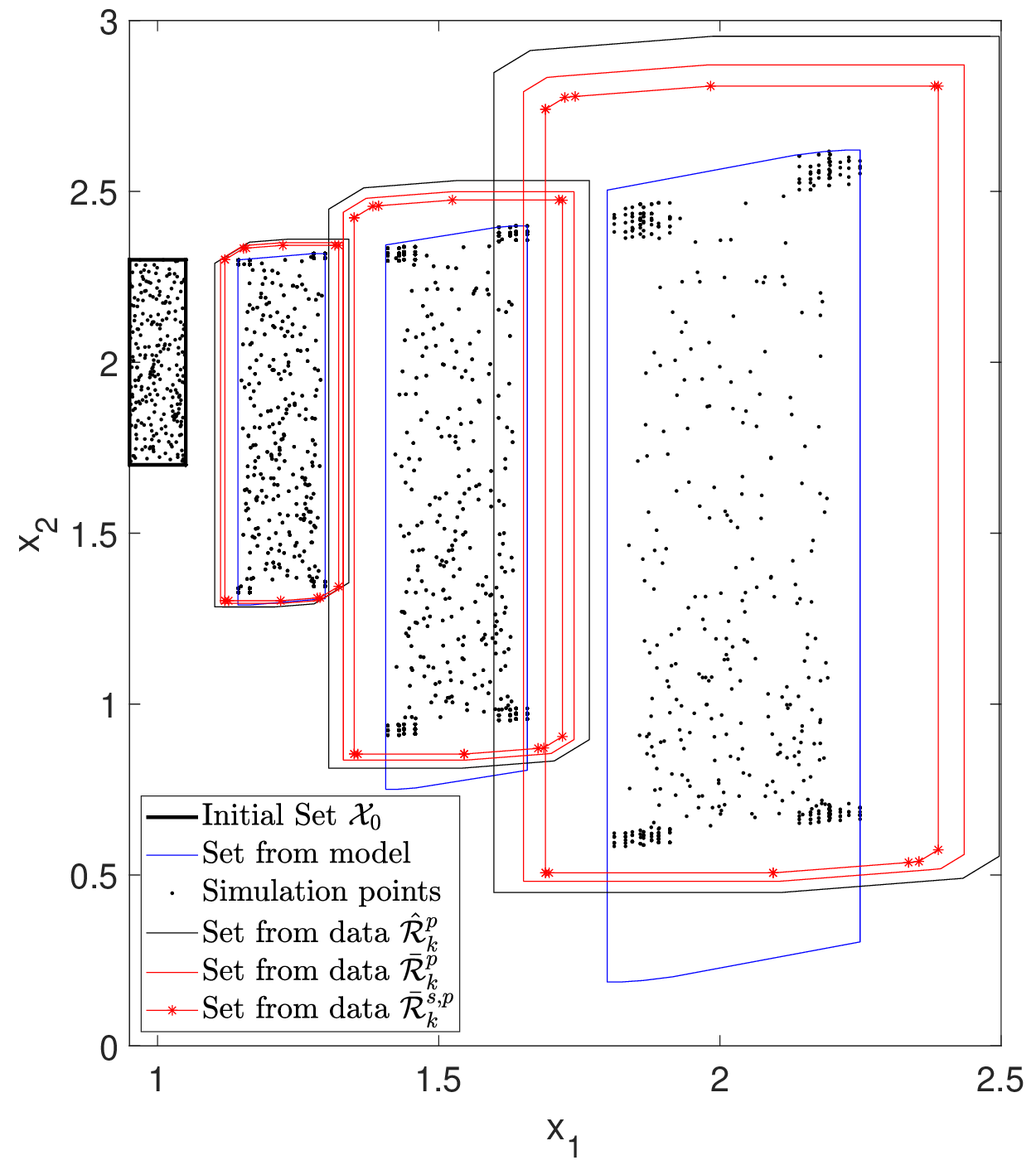}
        \caption{}
        \label{fig:poly_jor}
    \end{subfigure}
    \begin{subfigure}[h]{0.32\textwidth}
     \centering
        \includegraphics[scale=0.26]{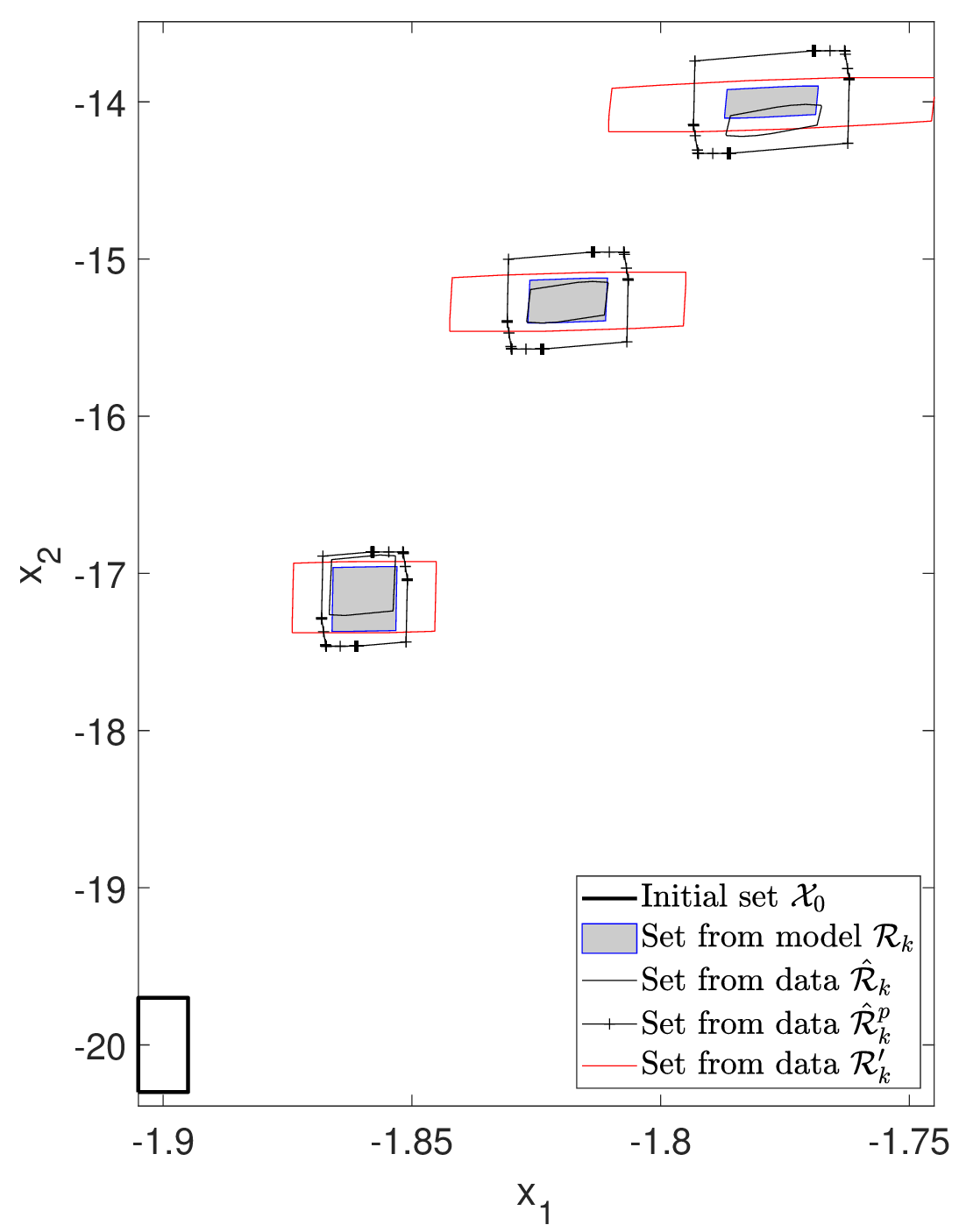}
        \caption{}
        \label{fig:singleex}
    \end{subfigure}
\caption{The projection of the reachable sets of the LTI system in \eqref{eq:sysexamplediscrete} from noisy input-state data with additional measurement noise computed via the proposed approaches in Propositions~\ref{prop:meas_zono_av} ($\hat{\mathcal{R}}_k^{\text{m}}$), Proposition~\ref{prop:meas_cmz_av} ($\bar{\mathcal{R}}_k^{\text{m}}$), 
  as well as the proposed practical approach from Algorithm~\ref{alg:LTIMeasReachability}  ($\tilde{\mathcal{R}}_k^{\text{m}}$) is presented in (a). The reachable sets in (b) of the polynomial system in \eqref{eq:pnonlinearexample} are computed using Algorithm~\ref{alg:PolyReachability} ($\hat{\mathcal{R}}^p_k$) and polynomial variants of Algorithm~\ref{alg:LTIConstrainedReachability} ($\bar{\mathcal{R}}^p_k$) and Algorithm~\ref{alg:LTISideInfoReachability} ($\bar{\mathcal{R}}_k^{\text{s},p}$). The reachable sets in (c) of the nonlinear tank system \cite{conf:nonlinearexample} from noisy measurement are computed using Algorithm~\ref{alg:LTIreach} ($\hat{\mathcal{R}}_k$), Algorithm~\ref{alg:PolyReachability} ($\hat{\mathcal{R}}^p_k$) and Algorithm~\ref{alg:LipReachability} ($\mathcal{R}_k^\prime$).}
    \label{fig:contreach}
     \vspace{-6mm}
\end{figure*}

\begin{itemize}
\item The true model based reachable sets $\mathcal{R}_k$.
\item The reachable sets $\hat{\mathcal{R}}_k^{\text{m}}$ using matrix zonotopes as introduced in Proposition~\ref{prop:meas_zono_av}.
\item The reachable sets $\bar{\mathcal{R}}_k^{\text{m}}$ using constrained matrix zonotopes as introduced in Proposition~\ref{prop:meas_cmz_av}.
\item The reachable sets $\tilde{\mathcal{R}}_k^{\text{m}}$ using the practical approach computed via Algorithm~\ref{alg:LTIMeasReachability}.
\end{itemize}
The data-driven reachable sets over-approximate the perfect model based reachable set correctly, and the practical approach provides the least conservative result. For further evaluation of the applicability of the practical approach, we additionally validate that
\begin{align*}
   \begin{bmatrix}A_{\text{tr}} & B_{\text{tr}}\end{bmatrix} \in \Bigg( M' \begin{bmatrix}Y_- \\ U_-\end{bmatrix} + \mathcal{M}_{AV} \Bigg)\begin{bmatrix}X_- \\ U_-\end{bmatrix}^\dagger.
\end{align*}
where $\mathcal{M}_{AV}$ is computed from $\mathcal{Z}_{AV}$ as described in \eqref{eq:Cmw} and \eqref{eq:Gmw}.
\subsection{Polynomial Systems}

We consider the problem of computing the reachable sets of the nonlinear discrete-time system described by
\begin{align}
    f_p(x,u) &= \begin{bmatrix}
    0.7x_1 + u_1 + 0.32x_1^2\\
    0.09x_1 + 0.32 u_2x_1 + 0.4 x_2^2
    \end{bmatrix}.
    \label{eq:pnonlinearexample}
\end{align}
The initial set is chosen to be $\mathcal{X}_0{=}\bigzono{\begin{bmatrix} 1 & 2 \end{bmatrix}^\t,\text{diag}(\begin{bmatrix} 0.05 & 0.3 \end{bmatrix})}$. The input set is $\mathcal{U}_k=\bigzono{\begin{bmatrix} 0.2 & 0.3 \end{bmatrix}^\t,\text{diag}(\begin{bmatrix} 0.01 & 0.02 \end{bmatrix})}$. We consider computing the reachable set when there is random noise sampled from the zonotope  $\mathcal{Z}_w=\bigzono{\begin{bmatrix} 0 & 0 \end{bmatrix}^\t,\begin{bmatrix}0.7 \times 10^{-4}& 0.7 \times 10^{-4}\end{bmatrix}^\t}$. We used as input data 140 data points ($T=140$) from 20 trajectories, of length seven. We present in Fig.~\ref{fig:poly_jor} the following reachable sets:
\begin{itemize}
    \item The reachable sets $\hat{\mathcal{R}}^p_k$ computed via Algorithm~\ref{alg:PolyReachability} using matrix zonotopes.
    \item The reachable sets $\bar{\mathcal{R}}^p_k$ using constrained matrix zonotopes computed via the polynomial variant of Algorithm~\ref{alg:LTIConstrainedReachability}.
    \item The reachable sets $\bar{\mathcal{R}}_k^{\text{s},p}$ utilizing the side information (similar to LTI example) computed via the polynomial variant of Algorithm~\ref{alg:LTISideInfoReachability}.
\end{itemize}

Due to the nonlinearity in the model, it is only possible using the state-of-art model-based reachability analysis techniques to compute an over-approximation of the exact reachable sets $\mathcal{R}_k$ \cite{conf:thesisalthoff}. Thus, it is acceptable to have the data-driven reachable set intersecting with the over-approximate model-based reachable set. We measured the execution time of the proposed algorithms as shown in Table~\ref{tab:execTimePoly}. Our approach in computing the set of monomials for the polynomial system using interval arithmetic is faster than the state-of-the-art nonlinear reachability analysis~\cite[p.18]{conf:cora}.

\begin{table}[h!]
\caption{Execution time in minutes for reachability analysis of the polynomial system.}
\label{tab:execTimePoly}
\centering
\normalsize
\begin{tabular}{c  c  }
\toprule
 Algorithm & Execution time \\
\midrule
Model-based polynomial reachability &$0.0458$  \\
Algorithm~\ref{alg:PolyReachability} &$4.695 \times 10^{-4}$ \\
Algorithm~\ref{alg:LTIConstrainedReachability} - polynomial version & $0.227$\\
Algorithm~\ref{alg:LTISideInfoReachability} - polynomial version & $0.287$\\
\bottomrule
\end{tabular}
\end{table}


\subsection{Lipschitz Nonlinear Systems}

We consider a scenario where we have collected data and we do not know the underlying system type. We apply the proposed data-driven reachability analysis to a continuous stirred tank reactor (CSTR) simulation model \cite{conf:nonlinearexample}. 
%
The initial set is a zonotope $\mathcal{X}_0 =\zono{\begin{bmatrix} -1.9 & -20 \end{bmatrix}^\t,\text{diag}(\begin{bmatrix} 0.005  & 0.3 \end{bmatrix})}$. The input set $\mathcal{U}_k =\zono{\begin{bmatrix}0.01 & 0.01 \end{bmatrix}^\t,\text{diag}(\begin{bmatrix} 0.1 & 0.2 \end{bmatrix})}$ and the noise set $\mathcal{Z}_w=\zono{0,\begin{bmatrix}5 \times 10^{-6} & 5 \times 10^{-6}\end{bmatrix}^\t}$. We show in Figure~\ref{fig:singleex} the following:
\begin{itemize}
    \item The model based reachable sets $\mathcal{R}_k$. 
    \item The reachable sets $\hat{\mathcal{R}}_k$ using matrix zonotopes in Algorithm~\ref{alg:LTIreach} for LTI system. The $\hat{\mathcal{R}}_k$ fails to over-approximate $\mathcal{R}_k$ as the system is nonlinear.
    \item The reachable sets $\hat{\mathcal{R}}^p_k$ using matrix zonotopes in Algorithm~\ref{alg:PolyReachability} for polynomial system. Approximating the underlying system using a polynomial is better than the LTI approximation. 
    \item The reachable sets $\mathcal{R}_k^\prime$ in Algorithm~\ref{alg:LipReachability} for Lipschitz nonlinear system. $\mathcal{R}_k^\prime$ provides theoretical guarantees and thus is more conservative than $\hat{\mathcal{R}}_k$ and $\hat{\mathcal{R}}^p_k$.
\end{itemize}

\begin{figure}
    \centering
    \includegraphics[scale=0.26]{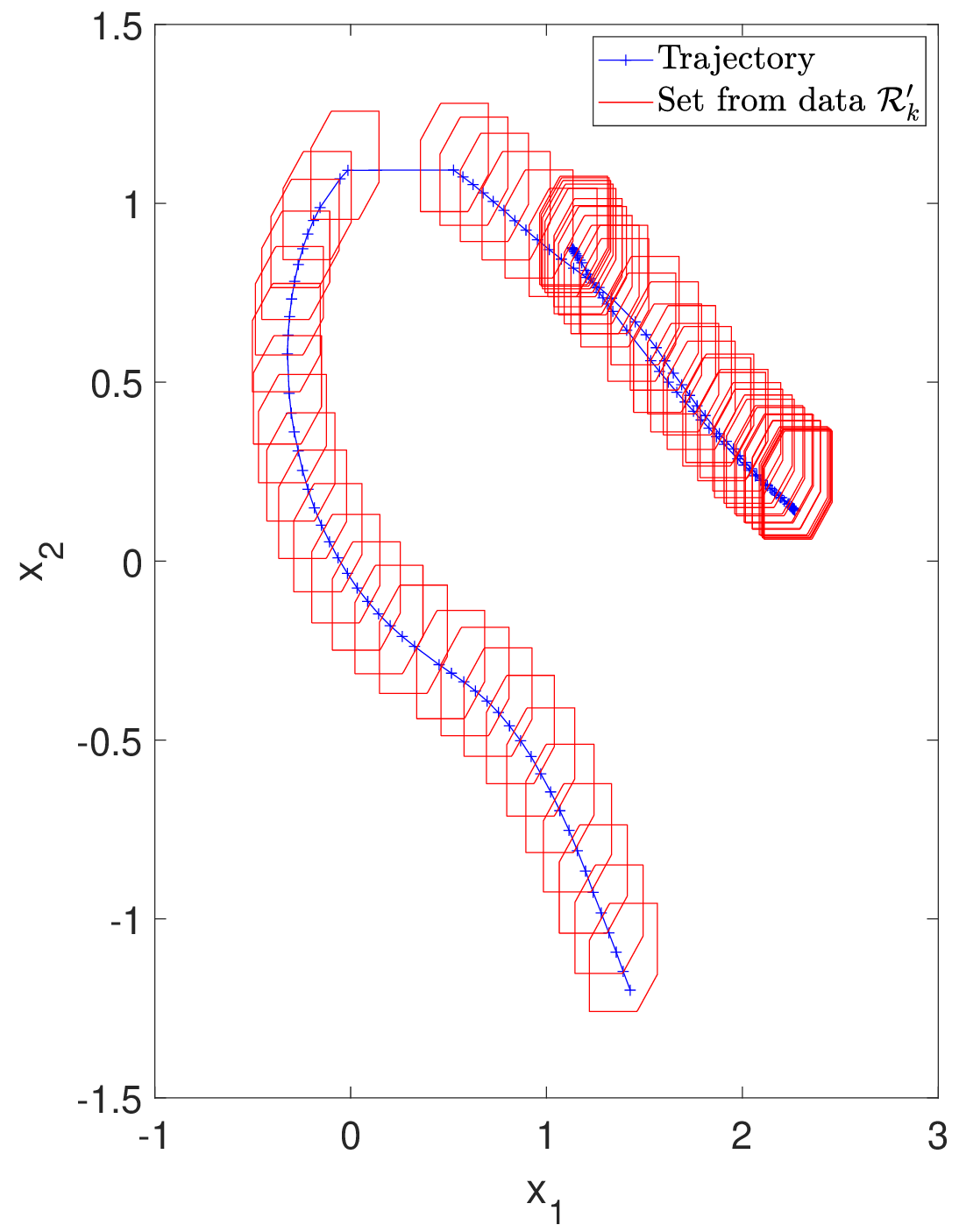}
    \caption{The computed reachable sets for the SVEA vehicle using Algorithm~\ref{alg:LipReachability}.}
    \label{fig:carlip}
    \vspace{-3mm}
\end{figure}
\subsection{Autonomous Vehicle}
%
We used in our experiments the SVEA (Small-Vehicles-for-Autonomy) vehicle \cite{Jiang2020a} shown in Fig.~\ref{fig:car}.  It is equipped with NVIDIA Jetson TX2 embedded computer and Qualisys motion capture system. We use historical data sets gathered from the same car from other driving scenarios than the presented ones. The input to the vehicle are the steering angle and the velocity and the output is the position of the vehicle. 
We consider process noise bounded by the zonotope $\mathcal{Z}_w=\zono{0,\begin{bmatrix}0.05& 0.05 \end{bmatrix}^\t}$. 
The reachable sets of a single step prediction using Algorithm~\ref{alg:LipReachability} are shown in Fig.~\ref{fig:carlip}. The reachable sets enclose the true trajectory.

\section{Conclusion and Future Work}\label{sec:con}

We consider the problem of computing the reachable sets directly from noisy data without requiring a mathematical model of the system. An approach to compute an over-approximation of the reachable set of the unknown system by over-approximating the reachable set of all sets of models consistent with the data and the noise bound is introduced. Further, we discuss some ideas for extending this result considering measurement noise added to the process noise. Moreover, we introduce a systematic approach to how prior information on the system can be included in the reachability analysis. Then, we provide algorithms to compute the reachable sets of polynomial systems given an upper bound on the degree of the polynomial and Lipschitz nonlinear systems, where we first fit a linear model and then over-approximate the model mismatch and Lagrange reminder from data.

It is part of ongoing work to investigate whether the proposed reachability analysis can be applied for (adaptive) robust model predictive control and how the resulting computational expenses and conservatism compare to the set-based parameter estimation and the respective forward-propagation with hyper-cubes \cite{conf:hypercude}, boxes \cite{conf:box} or more general set representations in \cite{conf:generalrep}. Furthermore, we want to quantify the amount of conservatism in our proposed approaches mathematically.

\section*{Acknowledgement}
The authors thank Prof. Marco Pavone for his useful insights and feedback on the proposed approaches. This work was supported by the Swedish Research Council, the Knut and Alice Wallenberg Foundation, the Deutsche Forschungsgemeinschaft (DFG, German Research Foundation) under Germany’s Excellence Strategy - EXC 2075 - 390740016, and the European Union's Horizon 2020 research and innovation programme under grant agreement No. 830927 (CONCORDIA project). The authors thank the International Max Planck Research School for Intelligent Systems (IMPRS-IS) for its support.

\begin{figure}[thp]
    \centering
    \includegraphics[scale=0.4]{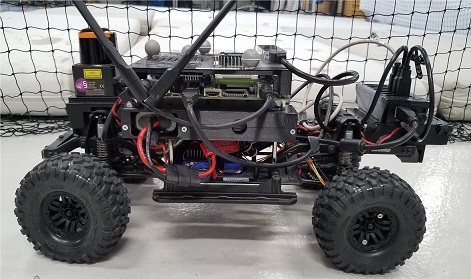}
    \caption{SVEA vehicle based on an NVIDIA Jetson TX2 embedded computer used to evaluate the proposed solutions.}
    \label{fig:car}
    \vspace{-6mm}
\end{figure}

%
%
\appendix
\section*{Proofs of Constrained Matrix Zonotopes Propositions}
\subsection{Proof of Proposition \ref{prob:minsum}}

\begin{proof}
The proof of \eqref{eq:RN} is along the lines of~\cite[Prop.~ 1]{conf:const_zono}. Let $\mathcal{N}_R$ be the right-hand side of \eqref{eq:RN}. The proof consists of proving two parts: 1) $R\mathcal{N}_1 \subseteq \mathcal{N}_R$ and 2) $\mathcal{N}_R \subseteq  R\mathcal{N}_1$. For any $X \in \mathcal{N}_1$, $\exists \beta^{(1:\gamma_{\mathcal{N}_1})}$, such that $X = C_{\mathcal{N}_1} + \sum_{i=1}^{\gamma_{\mathcal{N}_1}} \beta^{(i)} \, G_{\mathcal{N}_1}^{(i)}$ and $\sum_{i=1}^{\gamma_{\mathcal{N}_1}} \beta^{(i)} A_{\mathcal{N}_1}^{(i)} = B_{\mathcal{N}_1}$ and hence $RX = RC_{\mathcal{N}_1} + \sum_{i=1}^{\gamma_{\mathcal{N}_1}} \beta^{(i)}  R G_{\mathcal{N}_1}^{(i)}$. This implies that $RX \in \mathcal{N}_R$ by the definition of $\mathcal{N}_R$. Given that $X$ is arbitrary then $R\mathcal{N}_1 \subseteq \mathcal{N}_R$. Similarly, for any $X_R \in \mathcal{N}_R$, $\exists \beta^{(1:{\gamma_{\mathcal{N}_R}})}$, such that $X_R = R(C_{\mathcal{N}_1} + \sum_{i=1}^{\gamma_{\mathcal{N}_R}} \beta^{(i)} \, G_{\mathcal{N}_1}^{(i)})$ and $\sum_{i=1}^{\gamma_{\mathcal{N}_R}} \beta^{(i)} A_{\mathcal{N}_1}^{(i)} = B_{\mathcal{N}_1}$. it follows that $\exists X \in \mathcal{N}_1$ such that $X_R=RX$. Thus, $X_R \in R \mathcal{N}_1$ and therefore $\mathcal{N}_R \subseteq R \mathcal{N}_1$ as $X_R$ is arbitrary. We hence proved that $\mathcal{N}_R = R \mathcal{N}_1$.

Let $\mathcal{N}_{12}$ be the right-hand side of \eqref{eq:N1pN2} and let $X_1 \in \mathcal{N}_1$ and $X_2 \in \mathcal{N}_2$. Then, 
\begin{align*}
 \exists \beta_{\mathcal{N}_1}^{([1:\gamma_{\mathcal{N}_1}])}\!\!&:\! X_1 =\! C_{\mathcal{N}_1} + \sum_{i=1}^{\gamma_{\mathcal{N}_1}} \beta_{\mathcal{N}_1}^{(i)}  G_{\mathcal{N}_1}^{(i)},\quad \sum_{i=1}^{\gamma_{\mathcal{N}_1}} \beta_{\mathcal{N}_1}^{(i)} A_{\mathcal{N}_1}^{(i)} = B_{\mathcal{N}_1},
 \end{align*}
 and
 \begin{align*}
  \exists \beta_{\mathcal{N}_2}^{([1:\gamma_{\mathcal{N}_2}])}\!\!&:\! X_2 =\! C_{\mathcal{N}_2} + \sum_{i=1}^{\gamma_{\mathcal{N}_2}} \beta_{\mathcal{N}_2}^{(i)}  G_{\mathcal{N}_2}^{(i)},\quad \sum_{i=1}^{\gamma_{\mathcal{N}_2}} \beta_{\mathcal{N}_2}^{(i)} A_{\mathcal{N}_2}^{(i)} = B_{\mathcal{N}_2}.
\end{align*}
Let $\beta_{\mathcal{N}_{12}}^{([1:\gamma_{\mathcal{N}_{12}}])}=\begin{bmatrix}\beta_{\mathcal{N}_1}^{([1:\gamma_{\mathcal{N}_1}])}& \beta_{\mathcal{N}_2}^{([1:\gamma_{\mathcal{N}_2}])}\end{bmatrix}$. Then,
\begin{align*}
   X_1+X_2 =& C_{\mathcal{N}_{1}}+C_{\mathcal{N}_{2}} + \sum_{i=1}^{\gamma_{\mathcal{N}_{1}}} \beta_{\mathcal{N}_{12}}^{(i)}  G_{\mathcal{N}_{1}}^{(i)}+ \sum_{i=1}^{\gamma_{\mathcal{N}_{2}}} \beta_{\mathcal{N}_{12}}^{(\gamma_{\mathcal{N}_{1}}+i)}  G_{\mathcal{N}_{2}}^{(i)},\nonumber\\
   &\sum_{i=1}^{\gamma_{\mathcal{N}_{1}}} \beta_{\mathcal{N}_{12}}^{(i)} A_{\mathcal{N}_{1}}^{(i)} = B_{\mathcal{N}_{1}}, \sum_{i=1}^{\gamma_{\mathcal{N}_{2}}} \beta_{\mathcal{N}_{12}}^{(\gamma_{\mathcal{N}_{1}}+i)} A_{\mathcal{N}_{2}}^{(i)} = B_{\mathcal{N}_{2}}.
\end{align*}
Thus, $X_1+X_2 \in \mathcal{N}_{12}$ and therefore $\mathcal{N}_1+\mathcal{N}_2 \subseteq \mathcal{N}_{12}$. Conversely, let $X_{12} \in \mathcal{N}_{12}$, then 
\begin{align*}
\exists \beta_{\mathcal{N}_{12}}^{([1:\gamma_{\mathcal{N}_{12}}])}\!\!&:\! X_{12} =\! C_{\mathcal{N}_{12}} + \sum_{i=1}^{\gamma_{\mathcal{N}_{12}}} \beta_{\mathcal{N}_{12}}^{(i)}  G_{\mathcal{N}_{12}}^{(i)},\nonumber\\ &\sum_{i=1}^{\gamma_{\mathcal{N}_{12}}} \beta_{\mathcal{N}_{12}}^{(i)} A_{\mathcal{N}_{12}}^{(i)} = B_{\mathcal{N}_{12}}
 \end{align*}
 Partitioning $\beta_{\mathcal{N}_{12}}^{([1:\gamma_{\mathcal{N}_{12}}])}=\begin{bmatrix}\beta_{\mathcal{N}_1}^{([1:\gamma_{\mathcal{N}_1}])}& \beta_{\mathcal{N}_2}^{([1:\gamma_{\mathcal{N}_2}])}\end{bmatrix}$, it follows that there exist $X_1 \in \mathcal{N}_1$ and $X_2 \in \mathcal{N}_2$ such that $X_{12} = X_1 + X_2$. Therefore, $X_{12} \in \mathcal{N}_1 + \mathcal{N}_2$ and $ \mathcal{N}_{12} \subseteq \mathcal{N}_1+\mathcal{N}_2$.
 \end{proof}

\subsection{Proof of Proposition \ref{prop:cmzconstzonotope}}
\begin{proof}
  Let $\mathcal{C}_{1}$ be the right-hand side of \eqref{eq:matczono} and let $X_1 \in \mathcal{N}$ and $c \in \mathcal{C}$. Then, 
\begin{align*}
 \exists \beta_{\mathcal{N}}^{([1:\gamma_{\mathcal{N}}])}\!\!&:  X_1 =\! C_{\mathcal{N}} + \sum_{i=1}^{\gamma_{\mathcal{N}}} \beta_{\mathcal{N}}^{(i)}  G_{\mathcal{N}}^{(i)},\quad \sum_{i=1}^{\gamma_{\mathcal{N}}} \beta_{\mathcal{N}}^{(i)} A_{\mathcal{N}}^{(i)} = B_{\mathcal{N}}, \nonumber\\
  \exists \beta_{\mathcal{C}}^{([1:\gamma_{\mathcal{C}}])}\!\!&:  c =\! c_{\mathcal{C}} +  G_{\mathcal{C}} \beta_{\mathcal{C}}^{([1:\gamma_{\mathcal{C}}])}, \quad A_{\mathcal{C}}\beta_{\mathcal{C}}^{([1:\gamma_{\mathcal{C}}])}=b_{\mathcal{C}}.
 \end{align*}
Let 
\begin{align}
\beta_{\mathcal{C}_1}^{([1:\gamma_{\mathcal{C}_1}])}=\begin{bmatrix}\beta_{\mathcal{N}}^{([1:\gamma_{\mathcal{N}}])}&\beta_{\mathcal{C}}^{([1:\gamma_{\mathcal{C}}])}& 
\beta_{\mathcal{N}\mathcal{C}}^{([1:\gamma_{\mathcal{N}}\gamma_{\mathcal{C}}])}\end{bmatrix}. \label{eq:beta_c1_pp3}
\end{align} 
with
\begin{align*}
    \!\!&\beta_{\mathcal{N}\mathcal{C}}^{([1:\gamma_{\mathcal{N}}\gamma_{\mathcal{C}}])} = 
    \\ &\begin{bmatrix}\beta_{\mathcal{N}}^{(1)} \beta_{\mathcal{C}}^{(1)}  
    \dots \beta_{\mathcal{N}}^{(\gamma_{\mathcal{N}})}\beta_{\mathcal{C}}^{(1)}\!\!  &\!\! \beta_{\mathcal{N}}^{(1)}\beta_{\mathcal{C}}^{(2)}  \dots
    \beta_{\mathcal{N}}^{(\gamma_{\mathcal{N}})}\beta_{\mathcal{C}}^{(2)}  
    \dots  \beta_{\mathcal{N}}^{(\gamma_{\mathcal{N}})} \beta_{\mathcal{C}}^{(\gamma_{\mathcal{C}})}  \end{bmatrix}.
\end{align*}
Then, 
 \begin{align*}
 X_1\,c =& C_{\mathcal{N}} c_{\mathcal{C}} + \sum_{i=1}^{\gamma_{\mathcal{N}}} \beta_{\mathcal{C}_1}^{(i)}  G_{\mathcal{N}}^{(i)}  c_{\mathcal{C}} + C_{\mathcal{N}} G_{\mathcal{C}} \beta_{\mathcal{C}_1}^{([\gamma_{\mathcal{N}}+1:\gamma_{\mathcal{N}}+\gamma_{\mathcal{C}}])} \nonumber\\
 & + \sum_{i=1}^{\gamma_{\mathcal{N}}} \sum_{j=1}^{\gamma_{\mathcal{C}}} \beta_{\mathcal{N}}^{(i)} \beta_{\mathcal{C}}^{(j)} G_{\mathcal{N}}^{(i)}  g_{\mathcal{C}}^{(j)}. 
 \end{align*}
Next, we find the constraints on $\beta_{\mathcal{C}_1}^{([1:\gamma_{\mathcal{C}_1}])}$ in \eqref{eq:beta_c1_pp3}. The constraints on $\beta_{\mathcal{C}_{1}}^{([1:\gamma_{\mathcal{N}}])}$ and $\beta_{\mathcal{C}_{1}}^{([\gamma_{\mathcal{N}}+1:\gamma_{\mathcal{N}}+\gamma_{\mathcal{C}}])}$ can be captured by $ A_{\mathcal{N}} \beta_{\mathcal{C}_{1}}^{([1:\gamma_{\mathcal{N}}])}  = \mathrm{vec} (B_{\mathcal{N}})$ and $A_{\mathcal{C}}\beta_{\mathcal{C}_{1}}^{([\gamma_{\mathcal{N}}+1:\gamma_{\mathcal{N}}+\gamma_{\mathcal{C}}])}=b_{\mathcal{C}}$, respectively.

Finally, to find the constraint on $\beta_{\mathcal{N}\mathcal{C}}^{([1:\gamma_{\mathcal{N}}\gamma_{\mathcal{C}}])}$ (i.e. on $\beta_{\mathcal{C}_1}^{([\gamma_{\mathcal{N}}+\gamma_{\mathcal{C}}+1:\gamma_{\mathcal{C}_1}])}$), we first compute the intervals to which $\beta_{\mathcal{N}}^{([1:\gamma_{\mathcal{N}}])}$ and $\beta_{\mathcal{C}}^{([1:\gamma_{\mathcal{C}}])}$ are confined to:
\begin{align*}
\beta_{L,\mathcal{N}}^{(i)} \leq \beta_{\mathcal{N}}^{(i)} \leq \beta_{U,\mathcal{N}}^{(i)}, \quad
\beta_{L,\mathcal{C}}^{(j)} \leq \beta_{\mathcal{C}}^{(j)} \leq \beta_{U,\mathcal{C}}^{(j)}
\end{align*} in \eqref{eq:betaLZono2}-\eqref{eq:betaUCZono2}. Consequently, we know that
 \begin{align}
 \begin{split}
     \min ( 
     \beta_{L,\mathcal{N}}^{(i)}\beta_{L,\mathcal{C}}^{(j)},&\beta_{L,\mathcal{N}}^{(i)}\beta_{U,\mathcal{C}}^{(j)},\beta_{U,\mathcal{N}}^{(i)}\beta_{L,\mathcal{C}}^{(j)},\beta_{U,\mathcal{N}}^{(i)}\beta_{U,\mathcal{C}}^{(j)}
     ) \\
     &\leq  
     \beta_{\mathcal{N}}^{(i)}  \beta_{\mathcal{C}}^{(j)}
     \leq \\
     \max (
     \beta_{L,\mathcal{N}}^{(i)}\beta_{L,\mathcal{C}}^{(j)},&\beta_{L,\mathcal{N}}^{(i)}\beta_{U,\mathcal{C}}^{(j)},\beta_{U,\mathcal{N}}^{(i)}\beta_{L,\mathcal{C}}^{(j)},\beta_{U,\mathcal{N}}^{(i)}\beta_{U,\mathcal{C}}^{(j)}
     )
     \end{split}
     \label{eq:interval_correct}
 \end{align}
 holds for all $i=1,\dots,\gamma_{\mathcal{N}}$ and $j = 1, \dots, \gamma_{\mathcal{C}}$. This interval in \eqref{eq:interval_correct} can again be over-approximated by scaling the generator matrices $G_{\mathcal{N}}^{(i)} g_{\mathcal{C}}^{(j)}$ by $\max (\abs{\beta_{L,\mathcal{N}}^{(i)}\beta_{L,\mathcal{C}}^{(j)}},\abs{\beta_{L,\mathcal{N}}^{(i)}\beta_{U,\mathcal{C}}^{(j)}},\abs{\beta_{U,\mathcal{N}}^{(i)}\beta_{L,\mathcal{C}}^{(j)}},\abs{\beta_{U,\mathcal{N}}^{(i)}\beta_{U,\mathcal{C}}^{(j)}})$ and let $-1 \leq \beta_{\mathcal{C}_1}^{([\gamma_{\mathcal{N}}+\gamma_{\mathcal{C}}+1:\gamma_{\mathcal{C}_1}])} \leq 1$.
Thus, $X_1\,c \in \mathcal{C}_1$ or generally, $\mathcal{N}\,\mathcal{C} \subseteq \mathcal{C}_1$. 
\end{proof}


\bibliographystyle{ieeetr}        
\bibliography{ref} 
\vspace{-18mm}
\begin{IEEEbiography}[{\includegraphics[width=1in,height=1.25in,clip,keepaspectratio]{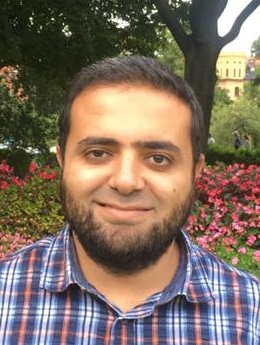}}]{Amr Alanwar} is an assistant professor at Constructor University, Germany. He received an M.Sc. in Computer Engineering from Ain Shams University, Cairo, Egypt, in 2013 and a Ph.D. in Computer Science from the Technical University of Munich, Germany, in 2020. He was a postdoctoral researcher at KTH Royal Institute of Technology. He was also a research assistant at the University of California, Los Angles. He received the Best Demonstration Paper Award at the 16th ACM/IEEE International Conference on Information Processing in Sensor Networks (IPSN/CPSWeek 2017) and was a finalist in the Qualcomm Innovation Fellowship for two years in a row.
\end{IEEEbiography}
\vspace{-18mm}
\begin{IEEEbiography}[{\includegraphics[width=1in,height=1.25in,clip,keepaspectratio]{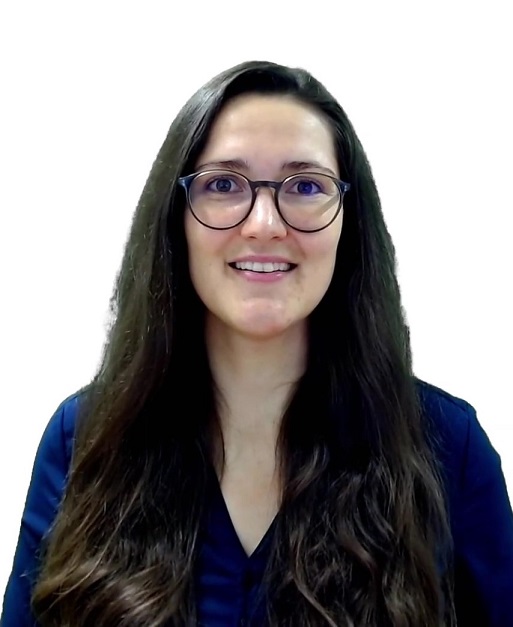}}]{Anne Koch (n\'ee Romer)}
received the M.Sc. in Engineering Science and Mechanics from the Georgia Institute of Technology, Atlanta, USA, in 2014, and the M.Sc. in Engineering Cybernetics from the University of Stuttgart, Germany, in 2016. In 2021, she received the Ph.D. from the Institute for Systems Theory and Automatic Control at the University of Stuttgart within the International Max Planck Research School for Intelligent Systems. Her re\-search interests include data-based systems analysis and controller design.
\end{IEEEbiography}
\vspace{-18mm}
\begin{IEEEbiography}[{\includegraphics[width=1in,height=1.25in,clip,keepaspectratio]{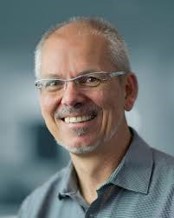}}]{Frank Allg\"{o}wer}
is a professor of mechanical engineering at the University of Stuttgart, Germany, and Director of the Institute for Systems Theory and Automatic Control (IST) there. Frank is active in serving the community in several roles: Among others he has been President of the International Federation of Automatic Control (IFAC) for the years 2017-2020, Vice-president for Technical Activities of the IEEE Control Systems Society for 2013/14, and Editor of the journal Automatica from 2001 until 2015. 
His research interests include predictive control, data-based control, cooperative control, and nonlinear control with application to a wide range of fields including systems biology.
\end{IEEEbiography}
\vspace{-18mm}
\begin{IEEEbiography}[{\includegraphics[width=1in,height=1.25in,clip,keepaspectratio]{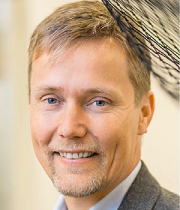}}]{Karl Henrik Johansson} is the Director of the Stockholm Strategic Research Area ICT The Next Generation and a Professor at the School of Electrical Engineering and Computer Science, KTH Royal Institute
of Technology. He received his M.Sc. and Ph.D. degrees from Lund University, Lund, Sweden. He has held visiting positions with the University of California, Berkeley, California Institute of Technology, Nanyang Technological University, HKUST Institute of Advanced Studies, and Norwegian University of Science and Technology. His research interests include networked control systems, cyber-physical systems, and applications in transportation, energy, and automation. 
\end{IEEEbiography}
\end{document}